\documentclass[aps,twocolumn,pra,10pt,superscriptaddress]{revtex4}

\usepackage{graphicx, xcolor}
\usepackage{hyperref}
\usepackage{physics}
\usepackage{color}
\usepackage{enumerate}
\usepackage{bm}
\usepackage[normalem]{ulem}
\usepackage{comment}

\usepackage[caption=false]{subfig}

\usepackage{graphicx}
\usepackage{braket}
\usepackage{amsmath}

\usepackage{amsthm}
\usepackage{amsfonts, amssymb}
\usepackage{bbm}
\usepackage{hyperref}

\newtheorem{corol}{Corollary}
\newtheorem{lemma}{Lemma}
\newtheorem{definition}{Definition}
\usepackage{mathtools}
\usepackage{mathrsfs}

\DeclarePairedDelimiterX{\infdivx}[2]{(}{)}{%
	#1\;\delimsize\|\;#2%
}

\begin{document}

\newtheorem{theorem}{Property}
\newtheorem{remark}{Remark}
\newtheorem{proposition}{Proposition}[section]

\theoremstyle{definition}
\newtheorem{assumptions}{Assumptions}[section]
\newtheorem{example}{Example}[section]
\newtheorem{rmk}{Remark}[section]
\newtheorem{conj}{Conjecture}[section]

\newcommand{\Det}{\text{Det}}
\newcommand{\ergo}{\mathcal{E}}
\newcommand{\deltaoff}{\Delta_{off}}
\newcommand{\en}{\mathfrak{E}}
\newcommand{\pass}{^{\downarrow}}
\newcommand{\state}{\hat\rho}
\newcommand{\ham}{\hat{H}}
\newcommand{\hilb}{\mathcal{H}}

\newcommand{\sigmax}{{\hat{\sigma}_x}}
\newcommand{\sigmay}{{\hat{\sigma}_y}}
\newcommand{\sigmaz}{{\hat{\sigma}_z}}

	\title{Zero-Fluctuation Quantum Work Extraction}

\author{Raffaele Salvia}
\email[]{raffaele.salvia@sns.it}
\affiliation{Scuola Normale Superiore, I-56127 Pisa, Italy}

\author{Vittorio Giovannetti}
\affiliation{NEST, Scuola Normale Superiore and Istituto Nanoscienze-CNR, I-56127 Pisa, Italy}

\begin{abstract}
We study the possibility of deterministic protocols for extracting work from quantum systems.
Focusing on the two-point measurement work extraction scenario, we prove that, with enough copies of the system, such zero-fluctuation protocols always exist if the Hamiltonian has a rational spectrum. Leveraging this result, we show that for any Hamiltonian, it is possible to construct an unitary driving protocol on sufficiently many copies of the system with work fluctuations strictly bounded within an arbitrary interval $\pm \delta$, albeit requiring exponentially many copies in $1/\delta$.\end{abstract}

	\maketitle

\section{Introduction}
The definition of thermodynamic work in quantum systems has been a longstanding question~\cite{Scovil1959, BochkovKuzovlev1977, Alicki1979, Skrzypczyk2014, Brandao2015, DeffnerReview2019, Ahmadi2023}: at microscopic scales, the work that can be extracted from a system acquires a stochastic nature, due to thermal~\cite{Jarzynski2011, Lahiri2021} and quantum fluctuations~\cite{Dahlsten2011, Horodecki2013, Aberg2013, Gemmer2015}.
A prominent paradigm for defining quantum work is the two-point measurement (TPM) scheme~\cite{twopoint1,twopoint2,twopoint3,twopoint4}, which provides a conceptually simple operational paradigm for defining fluctuating work in the quantum regime, and forms the basis for various quantum fluctuation theorems\cite{Talkner2007, Hanggi2015, Halpern2015, Lostaglio2018}.
The TPM protocol operates as follows: first, a projective energy measurement is performed on the system, collapsing its state into an energy eigenstate $\ket{\epsilon_i}$. Next, the system evolves unitarily under the action of an external agent. Finally, a second energy measurement reveals the final energy state $\ket{\epsilon_j}$ after evolution. By comparing the initial and final energies $\epsilon_i$ and $\epsilon_j$, one can define the stochastic work performed on the system as $w = \epsilon_i - \epsilon_j$. Repeating this TPM procedure many times generates a work probability distribution $\mathrm{P}(w)$.

A key optimization goal is to extract the maximum process average work $\langle W \rangle$ by suitably choosing the driving unitary $\hat{U}$. A fundamental bounds on $\langle W \rangle$ is provided by the \emph{ergotropy}~\cite{def_ergo, Alicki2013} - the energy difference between the initial state of the quantum system and its corresponding \emph{passive state}~\cite{Pusz1978, Lenard1978}.
More sophisticated protocols have been conceived in literature, that employ non-local operations to suppress the fluctuations of the extracted work $w$ around its average value $\langle W \rangle$~\cite{Renes2016, vanderMeer2017, Chubb2018, Friis2018, McKay2018}. 
In particular, \cite{Perarnau2018} showed that by collectively processing $n$ copies of a quantum system, the probability $\mathrm{P}\left(\lvert w - \langle W \rangle \rvert > \delta \right)$ of getting work fluctuations larger than a threshold $\delta > 0$ can be made to decay exponentially in $n$. The concept of ``$\epsilon$-deterministic'' work extraction has also been proposed, to indicate protocols that completely suppress fluctuations, except for a small failure probability~\cite{Horodecki2013, Salek2017}.

Our work goes beyond~\cite{Perarnau2018} to construct explicit protocols where the work fluctuations can be completely eliminated, such that the extracted work takes a single deterministic value $W^{\rm{(det)}}$. We prove that such protocols always exist for systems with rational spectra: if we collectively process multiple copies, then a finite $W^{\rm (det)} > 0$ can be extracted with zero fluctuations. For irrational spectra, for which a deterministic work extraction protocol may not be possible, we prove that with a sufficiently large number of copies we can always find a TPM work extraction protocol whose fluctuations can be strictly bound by an arbitrarily small constant (i.e., such that $\mathrm{P}\left(\lvert w - \langle W \rangle \rvert > \delta \right) = 0$).
Compared to the results of~\cite{Perarnau2018}, our zero-fluctuation protocols apply to a narrower regime of parameter space, but provide the strongest possible guarantee on work fluctuations by eliminating them completely. We identify permutations of energy levels between multiple copies that enable zero-fluctuation work extraction. For general spectra, we provide a stronger constraints on fluctuations than \cite{Perarnau2018} - ensuring that the extracted work $w$ from its expected value can be strictly bounded in a narrow band. The concept of work extraction with bounded fluctuations was firstly introduced in \cite{RichensMasanes2016}, in which the Authors show the existence of thermodynamic cycles with bounded fluctuation in some qubit and qutrit quantum systems; in this work, we generalize their results by providing a way to contruct bounded-fluctuations work extraction protocol for any system Hamiltonian $\ham$.
Our zero-fluctuation protocols could find applications in quantum heat engines or batteries where reliable work output is critical. The concept may also extend to bounding fluctuations of other quantities through global quantum operations, and help to introduce designs for stable quantum devices functioning in the finite copy regime.

The rest of manuscript is organised as follows.
In Sec.~\ref{nota-sec} we define the notation.
In Sec.~\ref{sec:formal_intro} we review the structure of TPM measurements,
and formally introduce the problem of finding the 
maximum amount of work $W^{\rm (det)}_{\max}(\hat{\rho}; \ham)$ that can be extracted deterministically from a quantum system described by an Hamiltonian $\ham$ initialised in the state $\hat\rho$.
In Sec.~\ref{sec:PRELIMINARIES}  we present some basic properties of the functional 
$W^{\rm (det)}_{\max}(\hat{\rho}; \ham)$, showing that it only depends on the spectrum of the
Hamiltonian, and on the occupancy levels of the input state $\hat{\rho}$. 
In Sec.~\ref{sec:superadditivity} we show that $W^{\rm (det)}_{\max}(\hat{\rho}; \ham)$  is super-additive in the number $n$ of copies of the system, and define the \emph{asymptotic maximum deterministic work-extraction rate} $\mathcal{R}(\mathcal{A}, \ham)$, which quantifies how much can be deterministically retrieved from a large $n\to\infty$ number of copies of the system. In Sec.~\ref{sec:upper_bound}, some upper bound for $\mathcal{R}(\mathcal{A}, \ham)$ are presented. Sec.~\ref{sec:examples} is devoted to the presentation of
same simple examples which are useful to shed light on the problem. 
In Sec.~\ref{sec:rational_spectra}, we show that if the eigenvalues of the system Hamiltonian $\ham$ are commensurable, then it is always possible to extract deterministically a non-zero amout of work for a sufficiently large number of copies of the system (i.e., $\mathcal{R}(\mathcal{A}, \ham) > 0$). Our proof is constructive, meaning that we provide an explicit protocol for deterministic work extraction for any Hamiltonian $\ham$ with a commensurable spectrum. 

Building upon this result, in Sec.~\ref{sec:spettri_irrazionali} we show that, by approximating a generic Hamiltonian $\ham$ to a $\delta$-close Hamiltonian $\ham^\prime$ with commensurable eigenvalues, we can construct (for a sufficiently large number of copies) a work extraction protocol whose fluctuations can be strictly bounded by an arbitrarily small constant $2\delta$.
In Sec.~\ref{sec:CLT_estimation} we provide another estimation of the asymptotic rate $\mathcal{R}(\mathcal{A}, \ham)$ using the local asymptotic normality of the distribution of energy eigenstates. This is not an upper bound neither a lower bound, but we heuristically expect it to be ``close'' to the actual value of $\mathcal{R}(\mathcal{A}, \ham)$ in most cases.
Conclusions are drawn in Sec.~\ref{sec:conclusions}. The paper also contains a couple of
technical Appendixes. 

\section{Notation} \label{nota-sec}

Consider a quantum system described by a $d$-dimensional Hilbert space ${\cal H}$, 
whose Hamiltonian
\begin{eqnarray}
\hat{H} := \sum_{i=0}^{M-1} \epsilon_{i}  \hat{\Pi}_i \label{defHAM}
\;, \end{eqnarray}  is characterized by
$M$  ($\leq \!\!d$) distinct eigenvalues $\{ \epsilon_{0}, \epsilon_{1}, \cdots, \epsilon_{M-1}\}$ of degeracies $\{ d_0, d_1,\cdots, d_{M-1}\}$, 
$\sum_{i=1}^{M-1} d_i =d$.  In the above expression the operators ${\hat{\Pi}}_0 ,{\hat{\Pi}}_1,\cdots, {\hat{\Pi}}_{M-1} $ form a complete set of orthonormal projectors
($\sum_{i=0}^{M-1}\hat{\Pi}_i=\hat{\openone}$,   $\hat{\Pi}_i \hat{\Pi}_{i'} = \delta_{i,i'}  \hat{\Pi}_i$)
associated  with the  energy eigenspaces 
${\cal H}_0 ,{\cal H}_1,\cdots, {\cal H}_{M-1}$ of $\hat{H}$ (${\cal H}=\oplus_{i=0}^{M-1}{\cal H}_i$, 
$\mbox{dim}{\cal H}_i =d_i$). 
Without loss of generality  we set equal to zero the  ground energy of the model and assume 
the following ordering for the spectral elements of $\hat{H}$, 
\begin{eqnarray} \begin{cases} 
\epsilon_0\;&=\;\; 0\;, \nonumber \\
 \epsilon_j \;&<\;\;  \epsilon_{j+1}\;, \qquad \forall j\in \{ 0,\cdots, M-2\}\;.
 \end{cases} 
  \end{eqnarray}  
  We also define the Linear, Completely Positive, Trace Preserving (LCPTP) channel  
\begin{eqnarray} \Phi(\cdots):= \sum_{i=0}^{M-1}  \hat{\Pi}_i \cdots \hat{\Pi}_i \;, \label{defphi} 
\end{eqnarray} 
 which 
induces full decoherence  with respect to the energy eigenspaces of the system.

    Given hence $\hat{\rho}$ an arbitrary quantum state of the system,
    we define 
    \begin{eqnarray} P(i|\hat{\rho}) := \mbox{Tr} [\hat{\Pi}_i\hat{\rho}]\;, \end{eqnarray}   the population it assigns to the $i$-th energy eigenspace ${\cal H}_i$ and   call  {\it non-zero energy level set} ${\mathbb{S}}[\hat{\rho}]$ 
    the set of  energy levels which  have a non-zero population, i.e. 
    \begin{eqnarray}
    {\mathbb{S}}[\hat{\rho}]: = \{ i: P(i|\hat{\rho})>0\}\;.
    \end{eqnarray} 
The energy diagonal counterpart of $\hat{\rho}$ obtained by the application of the  transformation $\Phi$, 
can be expressed as \begin{eqnarray} \label{diago} 
 \Phi(\hat{\rho}) := 
\sum_{i\in {\mathbb{S}}[\hat{\rho}] }  \hat{\Pi}_i \hat{\rho} \hat{\Pi}_i = \sum_{i\in {\mathbb{S}}[\hat{\rho}] } 
P(i|\hat{\rho})\;  \hat{\rho}_i \;,
\end{eqnarray} 
where for $i\in {\mathbb{S}}[\hat{\rho}]$,  
\begin{eqnarray}
\hat{\rho}_i &:=&\label{defi}  \hat{\Pi}_i \hat{\rho} \hat{\Pi}_i / P(i|\hat{\rho})
=  \sum_{k=0}^{r_i-1}  p_{i,k}  |\epsilon_{i,k}\rangle\langle \epsilon_{i,k}| \;,\end{eqnarray} 
 is the  projected
component of $\hat{\rho}$ on ${\cal H}_i$. In this expression $r_i$ represents the rank of  the matrix $\hat{\rho}_i$,  $p_{i,k}>0$ its 
non-zero eigenvalues, and  
 $|\epsilon_{i,k}\rangle\in {\cal H}_i$ the corresponding eigenvector. 
Notice that by construction one has that 
\begin{eqnarray} 
P(i|\hat{\rho}) = P(i|\Phi(\hat{\rho}))\;, \qquad  {\mathbb{S}}[\hat{\rho}]= {\mathbb{S}}[\Phi(\hat{\rho})]\;,
\end{eqnarray} 
and that 
the  support space of $\Phi(\hat{\rho})$
\begin{eqnarray}\label{nuova1} 
\mbox{Supp}[\Phi(\hat{\rho})] := \{ |\psi\rangle : \Phi(\hat{\rho}) |\psi\rangle \neq 0\}\;,
\end{eqnarray}   is
a proper subset of the direct sum of the energy egienspaces of the model over the elements of 
${\mathbb{S}}[\hat{\rho}]$. More precisely 
we can write 
\begin{eqnarray}\mbox{Supp}[\Phi(\hat{\rho})] &=&  \bigoplus_{i\in\mathbb{S}[\hat{\rho}] }{\cal H}_{i}[\Phi(\hat{\rho})] %
 \;,\label{directsum} \end{eqnarray} 
 where for $i\in\mathbb{S}[\hat{\rho}]$, 
 \begin{eqnarray} \label{defhispan} 
 {\cal H}_{i}[\Phi(\hat{\rho})]:= \mbox{Span}\{  |\epsilon_{i,k}\rangle; k=1,\cdots, r_i \}\subseteq  {\cal H}_{i}\;, \end{eqnarray} 
  represents the $r_i$ dimensional subset of ${\cal H}_i$ where 
 $\Phi(\hat{\rho})$ has no zero population (see Eq.~(\ref{defi})). 
  In case $\hat{H}$ is not degenerate (i.e. when $M=d$) then the inclusion in
 the last (\ref{directsum}) can be replaced by an identity 
implying that   $\mbox{Supp}[\Phi(\hat{\rho})]$  is fully characterized by the non-empty population index subset of  $\hat{\rho}$. 
For degenerate Hamiltonians such correspondence brakes since,  while it still true that   states $\hat{\rho}$ and $\hat{\varrho}$ whose diagonal ensembles have the same support share the same non-empty population index subset,
the opposite implication can be false (i.e. we can have ${\mathbb{S}}[\hat{\rho}]={\mathbb{S}}[\hat{\varrho}]$ but 
$\mbox{Supp}[\Phi(\hat{\rho})] \neq \mbox{Supp}[\Phi(\hat{\varrho})]$).

As it will be clear in the next sections, the support space~(\ref{directsum})  of the diagonal ensemble of a state plays a central role in our analysis. 
For this reason 
 given ${\cal A}$ a (non-empty) linear subset  of ${\cal H}$, we find it convenient to define ${\mathfrak{S}}_{\cal{A}}$  the set of density matrices $\hat{\rho}$ whose energy diagonal ensemble has support that corresponds to such space, i.e
\begin{eqnarray}\label{defsuppo} 
 {\mathfrak{S}}_{\cal{A}} &:=&\{  \hat{\rho}:   \mbox{Supp}[ \Phi(\hat{\rho})]= {{\cal A}}\}\;.
\end{eqnarray} 
By a closed inspection of Eq.~(\ref{directsum}) it turns out that only non-trivial (i.e. not empty) examples of ${\mathfrak{S}}_{\cal{A}}$ are those where ${\cal A}$ is 
a direct sum of a collection $\{ {\cal A}_0, {\cal A}_1,\cdots , {\cal A}_{M-1}\}$ of (possibly empty) linear subsets of the energy eigenspaces 
of the system Hamiltonian 
$\hat{H}$, i.e. 
\begin{eqnarray}
\label{defsuppo1} 
 {\cal A}&:=& \bigoplus_{i=0}^{M-1}{\cal A}_{i}  \;,\qquad 
 {\cal A}_{i}\subseteq {\cal H}_i \;. \end{eqnarray} 
 Notice also that while in general the elements of ${\mathfrak{S}}_{\cal{A}}$ could have different 
spectral decompositions, from Eq.~(\ref{directsum}) it follows that 
given 
\begin{eqnarray} \label{defroho} 
{\mathbb{S}}:= \{ i : \mbox{dim}[{\cal A}_{i}]>0\} \;,
\end{eqnarray} 
the set which identifies  the non-empty elements of $\{ {\cal A}_0, {\cal A}_1,\cdots , {\cal A}_{M-1}\}$,
we must have 
 \begin{eqnarray} \label{ide222}  \forall \hat{\rho}\in {\mathfrak{S}}_{\cal{A}} \Longrightarrow
 \left\{ \begin{array}{l}
 {\mathbb{S}}[\hat{\rho}]= {\mathbb{S}}\;, \\
 {\cal H}_{i}[\Phi(\hat{\rho})]= {\cal A}_i\;,  \qquad
   \forall i\in {\mathbb{S}}\;.
   \end{array} \right.
 \end{eqnarray} 
Special instances of the sums~(\ref{defsuppo1}) 
  are provided by the Hilbert space  itself ${\cal H} := \bigoplus_{i=0}^{M-1}{\cal H}_{i}$
  (in this case ${\mathfrak{S}}_{\cal{H}}$ includes all the states of the model), 
and by the single-state elements ${\cal A}^{[1,j]}:= 
\oplus_{i=1}^{M-1}{\cal A}^{[1,j]}_i$ characterized by the fact that 
their only not-trivial term  is the $j$-th one which corresponds to
a single not-null vector of 
the $j$-th energy eigenspace ${\cal H}_j$, so that the associated non-empty elements set is ${\mathbb{S}}= \{ j\}$ and 
\begin{equation} \label{enerankpure} 
\mbox{dim}[{\cal A}^{[1,j]}_{i}] = \delta_{j,i}\;.
\end{equation} 
Important examples  of density matrices which can be found in ${\mathfrak{S}}_{\cal{A}}$ are represented by the 
Gibbs-like states $\hat{\omega}_{\cal A}(\beta)$ 
obtained by taking a thermal state of inverse temperature~$\beta\geq 0$ and 
 filtering out the energy levels which are not in  ${\cal A}$, i.e.  
\begin{eqnarray}  \label{GIBBSlike} 
\hat{\omega}_{\cal A}(\beta) &: =& \frac{ \hat{\Pi}_{{\cal A}} e^{-\beta \hat{H} }}{Z_{{\cal A}}(\beta)  }  = \frac{ \sum_{i\in{\mathbb{S}}} \hat{\Pi}_{{\cal A}_i} e^{-\beta \epsilon_i}}{Z_{{\cal A}}(\beta)  } \;, 
\end{eqnarray} 
with $\hat{\Pi}_{{\cal A}_i}$ being
the projector on the $i$-th block ${\cal A}_{i}$ of ${\cal A}$  and with
\begin{eqnarray} 
Z_{{\cal A}}(\beta)  := \mbox{Tr}[  \hat{\Pi}_{{\cal A}} e^{-\beta \hat{H}} ]
= \sum_{i\in{\mathbb{S}}} e^{-\beta \epsilon_i} \mbox{Tr}[\hat{\Pi}_{{\cal A}_i}]\;,
\end{eqnarray}
where $\hat{\Pi}_{{\cal A}}= \sum_{i \in {\mathbb{S}}}  \hat{\Pi}_{{\cal A}_i}$ is the the projector on  ${\cal A}$.
We stress that by construction the states~$\hat{\omega}_{\cal A}(\beta)$ are invariant under $\Phi$, i.e.
 \begin{eqnarray} 
 \hat{\omega}_{\cal A}(\beta) = \Phi(\hat{\omega}_{\cal A}(\beta))\;.
 \end{eqnarray} 
 Notice also that in the high temperature limit $\beta=0$ , Eq.~(\ref{GIBBSlike}) reduces to the fully mixed state on ${\cal A}$, i.e. 
 \begin{eqnarray}\label{ddfadasdf} 
 \hat{\omega}_{\cal A}(0) :=\frac{ \hat{\Pi}_{{\cal A}} }{\mbox{Tr}[  \hat{\Pi}_{{\cal A}}] }\;, 
 \end{eqnarray} 
 which is still a proper element of ${\mathfrak{S}}_{\cal{A}}$.
On the contrary 
in the zero-temperature limit $\beta\rightarrow \infty$ of $\hat{\omega}_{\cal A}(\beta)$ 
Eq.~(\ref{GIBBSlike}) converges to a state which typically is not in 
${\mathfrak{S}}_{\cal{A}}$. Indeed  the latter corresponds to the 
 density matrix
\begin{eqnarray} \label{lmi} 
\lim_{\beta \rightarrow \infty} \hat{\omega}_{\cal A}(\beta) = \hat{\omega}_{{\cal A}_{\min}}(0) = \frac{ \hat{\Pi}_{{\cal A}_{\min}} }{\mbox{Tr}[  \hat{\Pi}_{{\cal A}_{\min}}] }\;,
\end{eqnarray} 
which has support on  the  restricted subspace ${\cal A}_{\min}:={\cal A}_{ \min_{i \in {\mathbb{S}}}}$ identified by
 the non-empty block term of  ${\cal A}$ that has 
the smallest energy eingenvalue, i.e.  \begin{eqnarray} 
\epsilon_{\min}({\cal A}) := \label{minimalenergyofA} 
 \min_{i \in {\mathbb{S}}} \epsilon_i =
  \epsilon_{\min_{i \in {\mathbb{S}}}}
  \;.
\end{eqnarray}
We finally introduce a partial ordering on the subspaces~(\ref{defsuppo1}):
\begin{definition}\label{defin0}
Given   two  direct sums
  of linear subsets of the energy eigenspace of the system,
${\cal A}:= \bigoplus_{i=0}^{M-1}{\cal A}_{i}$ and ${\cal A}':= \bigoplus_{i=0}^{M-1}{\cal A}'_{i}$, we
say that the former is {\rm not  dominated} by the latter
(in formulas ${\cal A} {\succeq} {\cal{A}}'$) if 
there exists a  energy preserving unitary mapping $\hat{V}$  that
maps each component of ${\cal A}$ into the corresponding element of ${\cal{A}}'$, i.e. 
\begin{eqnarray} \label{succdef1} 
&&{\cal A} \; {\succeq} \; {\cal{A}}'  \; \Longleftrightarrow \; \exists  \hat{V} \mbox{\rm unitary}, [\hat{H},\hat{V}] =0\;  \mbox{s.t.}  \\\nonumber
&& \qquad  {V}[{\cal A}_i] \subseteq {\cal A}'_i\;,\quad   \forall i\in\{ 0,\cdots, M-1\}  \;,
\end{eqnarray} 
with ${V}[{\cal A}_i]$ representing the image of ${\cal A}_i$ under the action of $\hat{V}$. 
In case the relation can also be inverted (i.e. if we also have ${\cal A}' {\succeq} {\cal{A}}$) we
say that the two sums are {\rm equivalent}  (in formula ${\cal A} \;{\sim}\; {\cal A}'$).
\end{definition}
Clearly a necessary and sufficient condition to have  that ${\cal A}$ is 
not  dominated by ${\cal A}'$ is that the sub-blocks of the former have  dimensions
which are not larger than  the corresponding ones of the latter, \begin{equation} 
{\cal A} \; {\succeq} \; {\cal{A}}'  \; \Longleftrightarrow \; 
\mbox{dim}[{\cal A}_{i}]  \leq  \mbox{dim}[{\cal A}'_{i}] \;, \quad \forall i\in \{0,\cdots, M-1\}\;.
\end{equation}
Similarly a necessary and sufficient condition to ensure that ${\cal A}$ and 
 ${\cal A}'$ are equivalent is instead given by 
 \begin{equation} 
{\cal A} \; {\sim} \; {\cal{A}}'  \; \Longleftrightarrow \; 
\mbox{dim}[{\cal A}_{i}]  =  \mbox{dim}[{\cal A}'_{i}] \;, \quad \forall i\in \{0,\cdots, M-1\}\;. 
\end{equation}
Observe also that  for all not trivial  ${\cal A}$ we can write
\begin{eqnarray} \label{succdef1MAXaqwe} 
{\cal A}^{[1,j]} {\succeq}  {\cal A} \; {\succeq} \; \bar{\cal A}\; 
{\succeq}  \; {\cal{H}} \;, 
\qquad  \forall j\in {\mathbb{S}}\;,
\end{eqnarray} 
where ${\mathbb{S}}$ is the  non-empty elements set of ${\cal A}$, 
 ${\cal A}^{[1,j]}$ is the single state subset defined in Eq.~(\ref{enerankpure}), and finally 
  $\bar{\cal A}$ is  the 
 direct sum obtained by replacing all
 non-empty elements of ${\cal A}$ with the associated energy eigenspaces
 of $\hat{H}$, i.e. 
 \begin{eqnarray} \label{inequimpo111} 
 \bar{\cal{A}}:= \bigoplus_{i=1}^{M-1} \bar{\cal{A}}_i\;, \quad 
\bar{\cal{A}}_i:= \left\{ \begin{array}{ll} 
{\cal{H}}_i & \forall i\in {\mathbb{S}}\;, \\ 
\varnothing & \forall i\notin {\mathbb{S}}\;.
\end{array} \right.
\end{eqnarray}

\section{Deterministic Work Extraction} 
\label{sec:formal_intro}
In the two-point measurement (TPM) formalism~\cite{twopoint1,twopoint2,twopoint3,twopoint4} the work we can extract from the state $\hat{\rho}$ of the system through the application of a unitary transformation $\hat{U}$ 
is determined through the following process.
At time $t_{\rm in}$,  before the application of $\hat{U}$, a projective measurement is performed w.r.t. to the energy projectors $\{ \hat{\Pi}_0,\hat{\Pi}_1,\cdots  \}$: following the formalism introduced in the previous section,  for 
each $i\in{\mathbb{S}}[\hat{\rho}]$ there is a non-zero probability $P(i|\hat{\rho})$ that 
 the system 
will be projected into the density matrix 
$\hat{\rho}_i$ of Eq.~(\ref{defi})
 hence setting  the input energy of the model at $E_{\rm in}=\epsilon_i$. The system is hence evolved through $\hat{U}$ and  
a second energy measurement is performed at time $t_{\rm out}$ obtaining the energy value $E_{\rm out}=\epsilon_j$ with probability 
\begin{eqnarray} P_{\hat{U}}(j| \hat{\rho}_i)&:=& \mbox{Tr} [\hat{\Pi}_j\hat{U}  \hat{\rho}_i\hat{U}^\dag] \nonumber\\ &=& 
\sum_{k=1}^{r_i} p_i^{(k)} \langle \epsilon_{i,k}|\hat{U}^\dag \hat{\Pi}_j \hat{U}|\epsilon_{i,k} \rangle \;. \label{defprel} 
\end{eqnarray}   The extracted work is described by the quantity
\begin{eqnarray}
w = E_{\rm in} - E_{\rm out} \;,
\end{eqnarray} 
which happens to be a random variable
that can take the discrete values $(\epsilon_i-\epsilon_j)$ with probabilities \begin{eqnarray}
P_{\hat{\rho};\hat{U}}(j,i) := P(i| \hat{\rho}) P_{\hat{U}}(j|\hat{\rho}_i) =   \mbox{Tr} [\hat{\Pi}_j \hat{U} \hat{\Pi}_i \hat{\rho}\hat{\Pi}_i\hat{U}^\dag ] 
\;, \label{probabilities_conproiettori}
\end{eqnarray} 
the corresponding distribution being formally described by the formula  
\begin{eqnarray} \label{distribution}
P^{(\hat{H})}_{\hat{\rho};\hat{U}}(w) := \sum_{j,i}  P_{\hat{\rho};\hat{U}}(j,i) \; \delta(w-(\epsilon_i-\epsilon_j))\;. 
\end{eqnarray} 
 It is important to stress that in the TPM protocol the unitary is fixed a priori and cannot be modified
after the acquisition of the first measurement outcome. It is clear that if we do allow for the possibility of adapting the unitary transformation to the measurement outcome we can recover much more energy than we get in the TPM protocol (indeed (at least for models where the Hamiltonian in not degenerate) we can recover the full amount of the energy stored into the system by simply using unitaries $\hat{U}_j$ which maps $|\epsilon_j\rangle$ into the ground state).
However in this way we are basically pumping entropy output of the system, which is equivalent to put the system in thermal contact with a zero-temperature bath.
Notice also that replacing  $\hat{\rho}$ with 
  its energy diagonal part ~(\ref{diago}) 
in the l.h.s. of
 Eq.~(\ref{probabilities_conproiettori}) the quantity  doesn't change (i.e. $P_{\hat{\rho};\hat{U}}(j,i)=P_{\Phi(\hat{\rho});\hat{U}}(j,i)$):
this implies that for what it concerns the work we can extract from the system via  TPM protocols, the states $\hat{\rho}$ and $\Phi(\hat{\rho})$ exhibit the same statistical properties, i.e. 
\begin{eqnarray}
P^{(\hat{H})}_{\hat{\rho};\hat{U}}(w)=P^{(\hat{H})}_{\Phi(\hat{\rho});\hat{U}}(w)\;.\end{eqnarray} 
Our focus is on the first momentum of this distribution, i.e. the quantity 
\begin{eqnarray}
\langle W_{\hat{U}}(\hat{\rho}; \hat{H}) \rangle &:=& 
\int d w P^{(\hat{H})}_{\hat{\rho};\hat{U}}(w) w =  \nonumber 
\sum_{j,i} P_{\hat{\rho};\hat{U}}(j,i)(\epsilon_i-\epsilon_j) \\
 &=& \mbox{Tr}[ \hat{H}\hat{\rho}] - 
\mbox{Tr}[ \hat{H}\hat{U} \Phi(\hat{\rho})\hat{U}^\dag] 
\nonumber\\  &=& \mbox{Tr}[ \hat{H}\Phi(\hat{\rho})] - 
\mbox{Tr}[ \hat{H}\hat{U} \Phi(\hat{\rho})\hat{U}^\dag]\;,   \label{ddfaver} 
\end{eqnarray} 
which represents the mean work we can extract from $\hat{\rho}$ (i.e. $\Phi(\hat{\rho})$)  when employing the unitary $\hat{U}$. 
Its  maximum value corresponds to the ergotropy ${\cal E}(\Phi(\hat{\rho}); \hat{H})$ of $\Phi(\hat{\rho})$~\cite{def_ergo, Alicki2013}, i.e. 
\begin{eqnarray} \label{upperbd} 
\langle W_{\max}(\hat{\rho}; \hat{H}) \rangle: = \max_{\hat{U}} \langle W_ {\hat{U}}(\hat{\rho}; \hat{H}) \rangle = {\cal E}(\Phi(\hat{\rho}); \hat{H})\;,
\end{eqnarray} 
with the optimal $\hat{U}$ which saturates the maximum being the transformation $\hat{U}_{\star}$ which transforms $\Phi(\hat{\rho})$ into its passive counterpart 
$\Phi(\hat{\rho})^{\downarrow}$~\cite{Pusz1978, Lenard1978}
(by the same token the  minimum of $\langle W_ {\hat{U}}(\hat{\rho}; \hat{H}) \rangle$ corresponds to  the anti-ergotropy  
${\cal E}_A(\Phi(\hat{\rho}); \hat{H})$ of the diagonal ensemble state). 
Notice that since the passive state energy is a Schur-concave functional~\cite{Alimuddin2020} it follows that 
${\cal E}(\Phi(\hat{\rho}); \hat{H})$  is always not larger than ${\cal E}(\hat{\rho};\hat{H})$, so that
\begin{eqnarray}
\langle W_ {\hat{U}}(\hat{\rho}; \hat{H}) \rangle  \leq {\cal E}(\Phi(\hat{\rho}); \hat{H})\leq {\cal E}(\hat{\rho}; \hat{H})\;, \qquad \forall \hat{U}\;,
\end{eqnarray} 
meaning that the TPM process  is less efficient than the ergotropy protocol in extracting energy form the state. 
We shall also consider  the variance of the extracted work, i.e. the quantity 
\begin{eqnarray}
\langle \Delta^2 W_ {\hat{U}}(\hat{\rho}; \hat{H}) \rangle &=& 
\int d w P^{(\hat{H})}_{\hat{\rho};\hat{U}}(w) \left[ w-\langle W_ {\hat{U}}(\hat{\rho}; \hat{H}) \rangle\right]^2 \nonumber  \\ 
 &=& \sum_{j,i} 
P_{\hat{\rho};\hat{U}}(j,i)\Big[(\epsilon_i-\epsilon_j)-\langle W_ {\hat{U}}(\hat{\rho}; \hat{H}) \rangle\Big]
^2\nonumber \\
&=& \langle W^2_{\hat{U}}(\hat{\rho};\hat{H}) \rangle -\langle W_ {\hat{U}}(\hat{\rho}; \hat{H}) \rangle^2\;. 
\end{eqnarray} 
Fist notice that for each system there exists always at least a choice of $\hat{U}$ such that 
$\langle \Delta^2 W_ {\hat{U}}(\hat{\rho}; \hat{H}) \rangle=0$ (for instance  $\hat{U}=\hat{I}$). If the associated mean value $W= \langle W_ {\hat{U}}(\hat{\rho}; \hat{H}) \rangle$ is  non negative (of course this not the case of 
$\hat{U}=\hat{I}$), we  say that 
for such unitaries the TPM protocol allows one to extract the work $W$ deterministically, i.e. with zero-fluctuations:
\begin{definition}\label{defin1}
A work value $W\geq 0$ is said to be deterministically extractable from  the state $\hat{\rho}$ of the system if there exists $\hat{U}$ unitary such that 
\begin{eqnarray}
P^{(\hat{H})}_{\hat{\rho};\hat{U}}(w)=\delta(w -W)\;,\end{eqnarray} 
or, equivalently,  if and only if 
 \begin{eqnarray}  \left\{ \begin{array}{l}  \langle  W_ {\hat{U}}(\hat{\rho}; \hat{H}) \rangle=W 
 \;,  \\ \langle \Delta^2 W_ {\hat{U}}(\hat{\rho}; \hat{H}) \rangle =0\;. 
 \end{array} \right. \label{defDETWORK} 
  \end{eqnarray} 
\end{definition} 
By looking carefully at the definitions we have introduced so far, it is clear the only possible values $W$  that fulfil Eq.~(\ref{defDETWORK}) 
 are those associated with the non-negative energy gaps of the spectrum of $\hat{H}$. 
More specifically we
can claim that 
a certain value of work $W\geq 0$ can be extracted deterministically from $\hat{\rho}$   if and only if there 
exists a mapping $\mu :  {\mathbb{S}}[\hat{\rho}] \mapsto \{ 0, 1,\cdots, M-1\}$ and a unitary evolution ${\hat{U}}$ such that 
\begin{eqnarray}  \forall i\in {\mathbb{S}}[\hat{\rho}] \qquad  
\left\{  \begin{array}{ll} \epsilon_i-\epsilon_{\mu(i)}&=W \;, \\ \\ 
P_{\hat{U}}(\mu(i)| \hat{\rho}_i) &=1\;.   \end{array} \right.   \label{def20} 
\end{eqnarray} 
 Furthermore invoking Eq.~(\ref{defprel}) we can recast the second condition in Eq.~(\ref{def20}) as
 \begin{eqnarray} \label{importantIMPP} 
 &&\sum_{k=1}^{r_i} p_i^{(k)} \langle \epsilon_{i,k}|\hat{U}^\dag \hat{\Pi}_{\mu(i)}  \hat{U}|\epsilon_{i,k} \rangle= 1  \\
 && \qquad \Longleftrightarrow
 \langle \epsilon_{i,k}|\hat{U}^\dag \hat{\Pi}_{\mu(i)}  \hat{U}|\epsilon_{i,k} \rangle= 1 \quad \forall k\in \{ 1,\cdots, r_i\}\;, \nonumber 
 \end{eqnarray} 
 where the second line follows from the fact  that the probabilities $p_i^{(k)}$ are all strictly positive.
Observe that the resulting expression is equivalent to say that the energy subspace ${\cal H}_{\mu(i)}$ must be sufficiently large to contain the full image of  the set ${\cal H}_i[\Phi(\hat{\rho})]$ defined in Eq.~(\ref{defhispan}).
We can hence equivalently write Eq.~(\ref{def20}) by saying
 that $W\geq 0$ can be extracted deterministically from $\hat{\rho}$   if and only if there 
exists a mapping $\mu :  {\mathbb{S}}[\hat{\rho}] \mapsto \{ 0, 1,\cdots, M-1\}$ and a unitary evolution ${\hat{U}}$ such that 
\begin{eqnarray}  \forall i\in {\mathbb{S}}[\hat{\rho}] \qquad  
\left\{  \begin{array}{ll} \epsilon_i-\epsilon_{\mu(i)}=W \;, \\ \\ 
 {U}\big[{\cal H}_i[\Phi(\hat{\rho})]\big] = {\cal H}_{\mu(i)}\;,   \end{array} \right.   \label{def243} 
\end{eqnarray}
with ${U}\big[{\cal H}_i[\Phi(\hat{\rho})]\big]$ being the image of ${\cal H}_i[\Phi(\hat{\rho})]$ under the action of $\hat{U}$. 
The above  expression can now be used to 
 establish the following general rules: 
\begin{lemma}\label{lemma111} Let 
  $\hat{U}$ be a unitary transformation
  which allows for the deterministic extraction of 
a work value $W\geq 0$  from the state $\hat{\rho}$.
Then such unitary
 will lead the
same outcome when applied to  any other  density matrix 
$\hat{\varrho}$ whose diagonal ensemble $\Phi(\hat{\varrho})$ has the same support
 of $\Phi(\hat{\rho})$, i.e. 
 \begin{eqnarray}\mbox{\rm Supp}[\Phi(\hat{\varrho})] =\mbox{\rm Supp}[\Phi(\hat{\rho})]
 \Longrightarrow
 \left\{ \begin{array}{l}  \langle  W_ {\hat{U}}(\hat{\varrho}; \hat{H}) \rangle=W 
 \;,  \\ \langle \Delta^2 W_ {\hat{U}}(\hat{\varrho}; \hat{H}) \rangle =0\;. 
 \end{array} \right.
  \end{eqnarray} 
\end{lemma}
\begin{proof} Since  the diagonal ensembles $\Phi(\hat{\varrho})$ and
$\Phi(\hat{\rho})$
have the same support it follows that 
 $\mathbb{S}[\hat{\varrho}]=\mathbb{S}[\hat{\rho}]$ and 
 ${\cal H}_{i}[\Phi(\hat{\varrho})] = {\cal H}_{i}[\Phi(\hat{\rho})]$ for all $i\in 
 \mathbb{S}[\hat{\rho}]$. Accordingly if the condition~(\ref{def243}) applies
 to $\hat{\rho}$ then it also applies to $\hat{\varrho}$.
 \end{proof}
It is worth stressing that Lemma~\ref{lemma111} does not requires $\Phi(\hat{\rho})$ and $\Phi(\hat{\varrho})$ to have the same spectrum: it only matters that 
they have the same support.

\begin{lemma}\label{lemma111222} Let $\hat{\rho}$ and $\hat{\rho}'$ be
two density matrices such that the support of 
  $\Phi(\hat{\rho}')$ can be mapped into the support of $\Phi(\hat{\rho})$ 
 via an energy preserving unitary operation. Then for each
  $\hat{U}$ unitary transformation
  which allows for the deterministic extraction of 
a work value $W\geq 0$  from the state $\hat{\rho}$, there exists
a new unitary $\hat{U}'$ which allows to do the same 
from $\hat{\rho}'$.
\end{lemma}
\begin{proof} Let $\hat{V}$ be the energy preserving unitary transformation 
that sends $\mbox{\rm Supp}[\Phi(\hat{\rho}')]$ into $\mbox{\rm Supp}[\Phi(\hat{\rho})]$. Recalling~(\ref{directsum}) this implies that 
for all $i \in \mathbb{S}[\hat{\rho}']$ we must have 
\begin{eqnarray}
{V}\big[{\cal H}_{i}[\Phi(\hat{\rho}')] \big] = {\cal H}_{i}[\Phi(\hat{\rho})] 
 \;,\label{directsum11}  \end{eqnarray} 
 where as usual we used ${V}\big[{\cal H}_{i}[\Phi(\hat{\rho}')] \big]$ to indicate the image
 of ${\cal H}_{i}[\Phi(\hat{\rho}')]$ under $\hat{V}$. 
The thesis hence follows by observing that if $\hat{U}$ fulfils the 
deterministic work extraction condition~(\ref{def243}) for $\hat{\rho}$, then the unitary 
 $\hat{U}' := \hat{U}\hat{V}$ does the same for $\hat{\rho}'$.
 \end{proof}
In the remaining of the present paper we shall focus on the characterization of the maximum work that can be deterministically extracted from 
a given input state:
\begin{definition}\label{defin2}
The Maximum Deterministic Extractable Work (or MDEW in brief) of a state 
 $\hat{\rho}$ is the maximum value of the values $W$ 
 which fulfil the  condition~(\ref{defDETWORK}), i.e. the quantity 
\begin{equation} \label{defmdew} 
W_{\max}^{(\rm det)}(\hat{\rho};\hat{H}) := \max_{\hat{U}} \{ \langle  W_ {\hat{U}}(\hat{\rho}; \hat{H}) \rangle  : \langle \Delta^2 W_ {\hat{U}}(\hat{\rho}; \hat{H}) \rangle =0\}\;.  
\end{equation} \end{definition} 
Clearly the configurations for which one expects  MDEW to be strictly positive correspond to {\it rare events}: this is a consequence of the fact that even the smallest perturbation in the spectrum of $\hat{H}$ or in the support of $\hat{\rho}$  will tend to assign a positive value to the TPM work variance functional $\langle \Delta^2 W_ {\hat{U}}(\hat{\rho}; \hat{H}) \rangle$ (for instance in the case of the  example of Eq.~(\ref{ide1}) discussed below, it is sufficient to take 
 $\epsilon_1=E$, $\epsilon_2=2E (1+\delta)$ with $\delta> 0$, or to add a small but non-zero population to the ground state of $\hat{\rho}$,  to get 
 $W_{\max}^{(\rm det)}(\hat{\rho};\hat{H})=0$). Nonetheless the study of $W_{\max}^{(\rm det)}(\hat{\rho};\hat{H})$ can 
 give us some hint on the efficiency of work extraction procedures in many cases of practical interests where
 geometrical or symmetry properties bound the system to assume assigned spectral characteristic. 
\section{Preliminary observations} \label{sec:PRELIMINARIES}

 It  also goes without mentioning that $W_{\max}^{(\rm det)}(\hat{\rho};\hat{H})$ coincides with 
$W_{\max}^{(\rm det)}(\Phi(\hat{\rho});\hat{H})$ and that,  thanks to~(\ref{upperbd}), it is  upper bounded by ${\cal E}(\Phi(\hat{\rho}); \hat{H})$, i.e.
\begin{eqnarray} \label{up1} 
W_{\max}^{(\rm det)}(\hat{\rho};\hat{H})= W_{\max}^{(\rm det)}(\Phi(\hat{\rho});\hat{H}) \leq {\cal E}(\Phi(\hat{\rho}); \hat{H})\;. 
\end{eqnarray} 
As a direct consequence of this fact, it follows that  if $\Phi(\hat{\rho})$ is a passive state, then 
$\langle W_{\max}(\hat{\rho}; \hat{H}) \rangle=0$ with the optimal unitary $\hat{U}_{\star}$ being the identity operator; accordingly we have
that the maximum of the deterministic work of these states is simply~$0$, i.e. 
\begin{eqnarray}\label{passiveMDEW} 
W_{\max}^{(\rm det)}(\hat{\rho};\hat{H})=0 \;, \qquad  \mbox{$\forall \Phi(\hat{\rho})$ passive.}
\end{eqnarray} 
Another case in which the MDEW can be easily computed is when 
 $\Phi(\hat{\rho})$ is pure, i.e. when  such state, and hence $\hat{\rho}$,  corresponds to the an eigenvector of $\hat{H}$: under this circumstance 
the maximum deterministic work we can get corresponds to the ergotropy which incidentally corresponds to the mean energy  of the  state, i.e. 
\begin{eqnarray}\label{ide1} 
W_{\max}^{(\rm det)}(\hat{\rho};\hat{H})={\Tr}[ \hat{H} \hat{\rho}]  \;, \qquad  \mbox{$\forall \Phi(\hat{\rho})$ pure.  }
\end{eqnarray} 
A less trivial example is provided by the following configuration:
let $\hat{H}= \epsilon_2 |\epsilon_2\rangle\langle \epsilon_2| +  \epsilon_1 |\epsilon_1\rangle\langle \epsilon_1|$ be a non-degenerate, three level Hamiltonian with uniforms energy gaps, i.e.  $\epsilon_1=E$, $\epsilon_2=2E$.
For any rank-2 density matrix $\Phi(\hat{\rho})$ with support space $\mbox{Span}\{ 
|\epsilon_1\rangle, |\epsilon_2\rangle\}$ we can then write 
\begin{eqnarray}\label{ide1bis} 
W_{\max}^{(\rm det)}(\hat{\rho};\hat{H})=
E \;,
\end{eqnarray} 
(the same holds if the matrix has rank-1 with non-zero population on $|\epsilon_1\rangle$, while if it has rank-1 but non-zero population on $|\epsilon_2 \rangle$ we get 
$W_{\max}^{(\rm det)}(\hat{\rho};\hat{H})=2E $). 
To see this observe  using the unitary $\hat{U}:= |\epsilon_1\rangle\langle \epsilon_2| +|\epsilon_0\rangle\langle \epsilon_1| + |\epsilon_2\rangle\langle \epsilon_0|$ 
we can induce the transitions  $|\epsilon_1\rangle\mapsto |\epsilon_0\rangle$ and 
$|\epsilon_2\rangle \mapsto |\epsilon_1\rangle$  which both yield exactly the work value $\Delta$.
To get more than this one would need necessarily to couple $|\epsilon_2\rangle$ with $|\epsilon_0\rangle$:  
 such amount of work however cannot be matched by any transitions that involves $|\epsilon_1\rangle$ as input state. As a result these type of operations will  involve random outcomes leading to 
non zero values of $\langle \Delta^2 W_ {\hat{U}}(\hat{\rho}; \hat{H}) \rangle$. Notice finally that, as a consequence of~Lemma~\ref{lemma111}, 
 (\ref{ide1}) holds true irrespectively from the specific values of the populations 
of the level $|\epsilon_2\rangle$ and $|\epsilon_1\rangle$.
This is a general rule that recalling the definitions of $\cal A$ and ${\mathfrak{S}}_{\cal{A}}$ 
introduced in Sec.~\ref{nota-sec} 
 can be summarized as follows:
\begin{corol}\label{cor1}   
All inputs states $\hat{\rho}$ of the set 
 ${\mathfrak{S}}_{\cal{A}}$ share the same MDEW value, i.e. 
 \begin{eqnarray}
 W_{\max}^{(\rm det)}(\hat{\rho};\hat{H})  =  {W}_{\max}^{(\rm det)} ({\cal A};\hat{H})\;, \quad \forall \hat{\rho}\in  \label{spe}
{\mathfrak{S}}_{\cal{A}}\;,
\end{eqnarray} 
where  recalling that  $\hat{\omega}_{{\cal A}}(0)$ of Eq.~(\ref{ddfadasdf}) belongs to ${\mathfrak{S}}_{\cal{A}}$ we can identify the constant 
${W}_{\max}^{(\rm det)} ({\cal A};\hat{H})$ as 
  \begin{eqnarray}
{W}_{\max}^{(\rm det)} ({\cal A};\hat{H}) :=  {W}_{\max}^{(\rm det)} (\hat{\omega}_{\cal A}(0);\hat{H})\;.   \label{spe222}
\end{eqnarray} 
Furthermore, irrespectively from the selected input state, such optimal value can be obtained using the same optimal unitary transformation $\hat{U}_{\star}$,  i.e. 
\begin{eqnarray} 
 \left\{ \begin{array}{l}  \langle  W_ {\hat{U}_{\star}}(\hat{\rho}; \hat{H}) \rangle=
 {W}_{\max}^{(\rm det)} ({\cal A};\hat{H}) 
 \;,  \\ \langle \Delta^2 W_ {\hat{U}_{\star}}(\hat{\rho}; \hat{H}) \rangle =0\;,
 \end{array} \right. \quad  \forall \hat{\rho}\in  \label{newddsds}
{\mathfrak{S}}_{\cal{A}}\;.
 \end{eqnarray} 
\end{corol} 
\begin{proof} Use Lemma~\ref{lemma111}  and  the fact that 
the elements of ${\mathfrak{S}}_{\cal{A}}$ share the same support 
space ${\cal A}$. 
\end{proof} 
In a similar way it follows that:  
\begin{corol}\label{cor2}
Let  ${\cal A}$ and ${\cal{A}}'$ be two equivalent (non-trivial) 
direct sums of linear subset of the energy eigenspaces of the system. Then 
the MDEW values ${W}_{\max}^{(\rm det)} ({\cal A};\hat{H})$ and 
 ${W}_{\max}^{(\rm det)} ({{\cal A}'};\hat{H})$  associated with  the
 states of the sets  ${\mathfrak{S}}_{\cal{A}}$ and ${\mathfrak{S}}_{\cal{A}'}$
 coincide, i.e. 
\begin{equation} 
{\cal A} \;{\sim}\; {\cal{A}}'  \; \Longrightarrow \; 
 {W}_{\max}^{(\rm det)} ({\cal A};\hat{H})=  {W}_{\max}^{(\rm det)} ({{\cal A}'};\hat{H})\;.
\end{equation} 
\end{corol}
\begin{proof} Use Lemma~\ref{lemma111222} and 
 the fact that according to Definition~\ref{defin0}
 the support spaces ${\cal A}$ and ${\cal A}'$ of 
the density matrices  of ${\mathfrak{S}}_{\cal{A}}$ and ${\mathfrak{S}}_{\cal{A}'}$ are connected by energy preserving unitary transformations that maps the first into the second and vice-versa.
 \end{proof}
The above results imply that, a part from the energy eigenvalues of $\hat{H}$,  the  MDEW value 
${W}_{\max}^{(\rm det)} ({\cal A};\hat{H})$ can  only depend upon 
the dimensions of the sub-blocks of ${\cal A}$. Accordingly we can always express 
${W}_{\max}^{(\rm det)} ({\cal A};\hat{H})$ as a function
${\cal W}(\vec{r}\left({\cal{A}}\right) , \vec{\epsilon})$
of the vectors $\vec{\epsilon}:= (\epsilon_0,\cdots, \epsilon_{M-1})$ and 
$\vec{r}\left({\cal{A}}\right) := (\mbox{dim}[{\cal A}_0],\cdots, \mbox{dim}[{\cal A}_{M-1}])$.  
As a special example  note that 
in the case of the single-state elements~(\ref{enerankpure}) from Eq.~(\ref{ide1}) we 
get 
\begin{equation} \label{identity-111} 
{W}_{\max}^{(\rm det)} ({\cal{A}}^{[1,j]};\hat{H}) = {\cal W}(\vec{r}[{\cal{A}}^{[1,j]}] , \vec{\epsilon}) = \vec{r}[{\cal{A}}^{[1,j]}]\cdot 
\vec{\epsilon} = \epsilon_i\;, 
\end{equation} 
while, recalling that  all passive states have maximum rank and hence belong to
${\mathfrak{S}}_{\cal{H}}$  we 
  can rewrite~(\ref{passiveMDEW}) as
\begin{eqnarray} \label{identity-1110} 
{\cal W}(\vec{r}\left( {\cal{H}}\right) , \vec{\epsilon}) = 0\;.
\end{eqnarray}
We next observe that the partial ordering~(\ref{succdef1}) introduced in Definition~\ref{defin0} can be used to rank the values of the function ${W}_{\max}^{(\rm det)} ({\cal A};\hat{H})$: 
\begin{lemma}\label{lemma33}
Let  ${\cal A}$ and ${\cal{A}}'$ be two (non-trivial) 
direct sums of linear subset of the energy eigenspaces of the system.
If ${\cal A}$ is   not  dominated by ${\cal{A}}'$ then 
the MDEW value ${W}_{\max}^{(\rm det)} ({\cal A};\hat{H})$
 is larger than or equal to 
 ${W}_{\max}^{(\rm det)} ({\cal A}';\hat{H})$,
i.e.
\begin{equation} 
{\cal A} \; {\succeq} \; {\cal{A}}' 
 \quad \Longrightarrow\quad   
{W}_{\max}^{(\rm det)} ({\cal A};\hat{H}) \geq   {W}_{\max}^{(\rm det)} ({\cal A}';\hat{H})\;.
\end{equation} 
\end{lemma}
\begin{proof} According to (\ref{succdef1}) there exists an energy-preserving unitary transformation $\hat{V}$, that maps the $i$-th subspace of ${\cal A}$ into the corresponding
one of ${\cal A}'$.
 Let now $\hat{U}'_{\star}$ be the optimal unitary map which applied to a generic states of ${\mathfrak{S}}_{\cal{A}'}$ enable us to extract the work value ${W}_{\max}^{(\rm det)} ({\cal A}';\hat{H})$ from the system. The thesis then follows by observing that  the unitary $\hat{U}_{\star}' \hat{V}$ applied to the elements of 
  ${\mathfrak{S}}_{\cal{A}}$ enable the deterministic extraction of the work level 
  ${W}_{\max}^{(\rm det)} ({\cal A}';\hat{H})$, which hence, by construction is a lower bound of the MDEW we can get from ${\mathfrak{S}}_{\cal{A}}$.
\end{proof} 
In particular from (\ref{succdef1MAXaqwe}) we get the following bounds 
\begin{equation} \epsilon_i \label{impoiiii} 
 \geq {W}_{\max}^{(\rm det)} ({\cal A};\hat{H}) \geq   {W}_{\max}^{(\rm det)} (\bar{\cal A};\hat{H})\geq 0\;, \quad \forall i\in {\mathbb{S}}\;,
\end{equation}  
where in writing the leftmost and rightmost terms  we used the identities (\ref{identity-111}) and
(\ref{identity-1110}) respectively.  
\\

The dependence of ${W}_{\max}^{(\rm det)} ({\cal A};\hat{H})$ with respect to the spectrum of
$\hat{H}$ for fixed choices of ${\cal A}$ is slightly more involved and, as will be discussed in
Sec.~\ref{sec:examples} 
 can lead to unexpected results. 
\section{Super-additivity properties and the asymptotic MDEW ratio}  
\label{sec:superadditivity}

Consider next the case where we have $n$-copies of the input state $\hat{\rho}$ for a system where the global Hamiltonian is
composed by a sum $\hat{H}^{(n)}: =\sum_{k=1}^n \hat{H}_k$ of homogeneous local terms ($\hat{H}_k$ being the local Hamiltonian of the $k$-th copy).
We are interested in determining how the $n$-copies MDEW, i.e. the quantity $W_{\max}^{(\rm det)}(\hat{\rho}^{\otimes n};\hat{H}^{(n)})$,
scales with~$n$. Let us start with some preliminary observations. First of all, notice that, due to the absence of 
interaction among the various copies of the system, 
the $n$-uses energy decoherence  LCPTP map of the model correspond to the $n$ copies of the map $\Phi$ of Eq.~(\ref{defphi}), i.e. $\Phi^{(n)} = \Phi^{\otimes n}$.
From this it hence follows that 
 if ${\cal A}$ is the support space of
$\Phi(\hat{\rho})$ then ${\cal A}^{\otimes n}$ is the support of $\Phi^{(n)}(\hat{\rho}^{\otimes n})=\Phi(\hat{\rho})^{\otimes n}$, i.e. 
\begin{eqnarray} 
\hat{\rho} \in {\mathfrak{S}}_{\cal{A}} \Longrightarrow \hat{\rho}^{\otimes n} \in {\mathfrak{S}}_{{{\cal A}}^{\otimes n}}\;. 
\end{eqnarray} 
From  Eq.~(\ref{spe}) and (\ref{spe222}) we can thus conclude that 
\begin{equation} 
W_{\max}^{(\rm det)}(\hat{\rho}^{\otimes n};\hat{H}^{(n)})= 
{W}_{\max}^{(\rm det)} ({\cal A}^{\otimes n};\hat{H}^{(n)})\ \;, \quad \forall \hat{\rho}\in  {\mathfrak{S}}_{\cal{A}}\;,
\end{equation} 
with 
\begin{eqnarray} 
{W}_{\max}^{(\rm det)} ({\cal A}^{\otimes n};\hat{H}^{(n)}) 
&=& {W}_{\max}^{(\rm det)} (\hat{\omega}_{{\cal A}^{\otimes n}}(0);\hat{H}^{(n)}) \nonumber \\
&=&{W}_{\max}^{(\rm det)} (\hat{\omega}^{\otimes n}_{\cal A}(0);\hat{H}^{(n)})\;,
\end{eqnarray}
where in the second line we used the identity $\hat{\omega}_{{\cal A}^{\otimes n}}(0)=\hat{\omega}^{\otimes n}_{\cal A}(0)$.
We can then arrive to the following inequality 
\begin{equation} 
 {W}_{\max}^{(\rm det)} ({\cal A}^{\otimes n};\hat{H}^{(n)})\geq n  {W}_{\max}^{(\rm det)} ({\cal A};\hat{H})\label{ien} \;,
\end{equation} 
by observing that if there exists a unitary procedure that extracts deterministic work $W_{\max}^{(\rm det)}({\cal A};\hat{H})$  from a single copy of a state (say $\hat{\omega}_{{\cal A}}(0)$)  we can simply
reiterate it to extract $n$ times such quantity from $n$ copy of the same density matrix (i.e. from ~$\hat{\omega}^{\otimes n}_{{\cal A}}(0)$).  
On the contrary there are examples which show that the gap in Eq.~(\ref{ien}) is non zero. 
For instance adding an extra energy level $|-1\rangle$  with energy  $-\delta$   to the example of Eq.~(\ref{ide1}), 
it turns out that as long as $\delta$ is positive and  $\neq \Delta$,
from $\hat{\rho}^{\otimes 2}$ we can extract energy $
2 \Delta  +\delta  > 2 \Delta$
which is larger than twice the max value we can get from a single copy of $\hat{\rho}$.
Using the same argument we can also conclude that for all $n$, $k$ integer the following super-additivity rule holds,  
\begin{eqnarray} \label{defn} 
&&\!\!\!\!\!\! \!\!\!\!\!\! {W}_{\max}^{(\rm det)} ({\cal A}^{\otimes (n+k)};\hat{H}^{(n+k)}) \geq\label{iene}  \\
&&\;\;\;\; {W}_{\max}^{(\rm det)} ({\cal A}^{\otimes n};
\hat{H}^{(n)})+{W}_{\max}^{(\rm det)}({\cal A}^{\otimes k};\hat{H}^{(k)})\;.\nonumber 
\end{eqnarray} 
A slightly less trivial observation is that there exist models for which even though 
${W}_{\max}^{(\rm det)} ({\cal A};\hat{H})=0$, 
  for sufficiently large $n$
one has
${W}_{\max}^{(\rm det)} ({\cal A}^{\otimes n};
\hat{H}^{(n)})>0$: we from now on we shall call this 
 {{\it strong-super-additivity}} property of the maximum deterministic TPM work. 
 Motivated by this observation
we define the asymptotic MDEW ratio as 
\begin{eqnarray} \label{impo11} 
{\cal R}({\cal A};\hat{H}) := \limsup_{n\rightarrow \infty}{\cal R}_n({\cal A};\hat{H})\;,
\end{eqnarray} 
with 
\begin{eqnarray} 
\label{def_Rn}
{\cal R}_n({\cal A};\hat{H})&:=& \frac{{W}_{\max}^{(\rm det)}({\cal A}^{\otimes n};\hat{H}^{(n)})  }{n} \nonumber \\
&=&{W}_{\max}^{(\rm det)}({\cal A}^{\otimes n};\hat{H}^{(n)}/n) \;. 
\end{eqnarray} 
From Eq.~(\ref{iene}) it follows that ${\cal R}_n({\cal A};\hat{H})$
 while not necessarily monotonically increasing is weakly increasing~\cite{wi}, meaning that, even if oscillating it still admits a proper $n\rightarrow \infty$ limit, i.e. 
\begin{eqnarray} \label{impo11bis} 
 \lim_{n\rightarrow \infty}{\cal R}_n({\cal A};\hat{H}) &=& \limsup_{n\rightarrow \infty}{\cal R}_n({\cal A};\hat{H})\\\nonumber 
  &=&
{\cal R}({\cal A};\hat{H}) := \max_{n} {\cal R}_n({\cal A};\hat{H})\;. 
\end{eqnarray} 

\subsection{Upper bounds} 
\label{sec:upper_bound}
A natural upper bound for ${\cal R}({\cal A};\hat{H})$ (and hence for all ${\cal R}_n({\cal A};\hat{H})$) is provided by 
 the minimal energy eigenvalue $\epsilon_{\min}({\cal A})$ of Eq.~(\ref{minimalenergyofA}) associated with  ${\cal A}$, i.e.
 \begin{eqnarray}{\cal R}({\cal A};\hat{H}) \leq \epsilon_{\min}({\cal A})   \;. \label{uppernew}
 \end{eqnarray}  
This formally follows from Eq.~(\ref{impoiiii}) by taking the minimum with respect to all possible choices
of $\epsilon_i$. More intuitively the bound~(\ref{uppernew}) can be explained by 
 that fact that i)  for all $n$, $\hat{\Pi}_{\cal A}^{\otimes n}$
 has a non-zero overlap with  the $n$-fold
copy of such level, and ii) we cannot extract more than $n \epsilon_{\min}({\cal A})$ energy from such configuration. 
Equation~(\ref{impo11}) establishes that the results of~\cite{Perarnau2018}  cannot be used to provide a full characterization of 
${\cal R}({\cal A};\hat{H})$.  Notice also that 
as a consequence of (\ref{uppernew})  it follows that 
if  ${\mathbb{S}}\left[ {\cal{A}}\right]$ contains the ground state energy level then the associated asymptotic ratio is zero, i.e.
\begin{eqnarray}\label{nogo}  0\in {\mathbb{S}}\; \Longrightarrow \; 
{\cal R}({\cal A};\hat{H}) ={\cal R}_n({\cal A};\hat{H}) =0\;. \end{eqnarray}

An improvement w.r.t.~(\ref{uppernew}) can be obtained 
 invoking~~(\ref{up1}), which for   $\hat{\rho}\in{\mathfrak{S}}_{\cal{A}}$ allows us to write 
\begin{eqnarray} 
&&{\cal R}_n({\cal A};\hat{H}) = \tfrac{{W}_{\max}^{(\rm det)} (\hat{\rho}^{\otimes n};\hat{H}^{(n)})}{n} 
\leq \tfrac{ {\cal E}(\Phi(\hat{\rho})^{\otimes n}; \hat{H}^{(n)})}{n}\nonumber 
\\  
&&\quad  \leq 
\limsup_{n\rightarrow \infty} \tfrac{ {\cal E}(\Phi(\hat{\rho})^{\otimes n}; \hat{H}^{(n)})}{n} =:
 {\cal E}_{\rm tot} (\Phi(\hat{\rho}); \hat{H})\;,
\end{eqnarray} 
where ${\cal E}_{\rm tot} (\Phi(\hat{\rho}); \hat{H})$ is the total ergotropy of the state $\Phi(\hat{\rho})$~\cite{def_ergo, Alicki2013}. 
Taking the minimum of the last term over all possible choices of $\hat{\rho}$, and taking the $n\rightarrow \infty$ limit, finally allows us to write
\begin{eqnarray} \label{newbound11}          
{\cal R}_n({\cal A};\hat{H})  &\leq & \min_{\hat{\rho}\in{\mathfrak{S}}_{\cal{A}} }  {\cal E}_{\rm tot} (\Phi(\hat{\rho}); \hat{H})\;.
\end{eqnarray} 
Recall next that the Gibbs-like states~(\ref{GIBBSlike}) are special instance of elements of 
${\mathfrak{S}}_{\cal{A}}$: therefore a simplified, yet in principle less performant, version of ~(\ref{newbound11}) 
 is given by 
 \begin{eqnarray} \label{newbound2} 
{\cal R}({\cal A};\hat{H}) &\leq & \min_{\beta >0}  {\cal E}_{\rm tot} (\hat{\omega}_{\cal A}(\beta); \hat{H})\;. 
\end{eqnarray} 
Written in this form it is now easy to verify that~(\ref{newbound2}) (and hence~(\ref{newbound11}))
implies~(\ref{uppernew}): indeed taking the limit for $\beta \rightarrow \infty$ 
and invoking Eq.~(\ref{lmi}) we can claim that
${\cal R}({\cal A};\hat{H})$ is upper bounded by the total ergotropy of $\hat{\omega}_{{\cal A}_{\min}}(0)$
which in turns cannot be larger than mean energy $\epsilon_{\min}({\cal A})$ of such a state. 
Most importantly, 
as shown in Appendix~\ref{deri}, at least in the case in which $\hat{H}$ has no degenerate spectrum on ${\cal A}$ (i.e. when for all $i\in {\mathbb{S}}$ the projectors $\hat{\Pi}_{{\cal A}_i}$ are rank-one operators), it is possible
to show that  the r.h.s. of~(\ref{newbound11}) and (\ref{newbound2}) coincide. 
Furthermore in  Appendix~\ref{sec:fixed_point} we show that the value of $\beta$ that realizes the minimum~(\ref{newbound2}) satisfies the special property
\begin{eqnarray}
\label{fixed_point}
S((\hat{\omega}_{\cal A}(\beta)) = S(\hat\tau_\beta) \;, \qquad \hat\tau_\beta := \tfrac{e^{-\beta\ham}}{ \Tr[e^{-\beta\ham}] } \;,
\end{eqnarray}
where $S(\cdots):= \mbox{Tr}[ (\cdots) \log (\cdots)]$ is the von Neumann entropy functional and 
$\hat\tau_\beta$ is the thermal Gibbs state of the model with inverse temperature $\beta$. 
A comparison between  $\mathcal{R}_{n}( \mathcal{A} ; \ham)$ and the upper bound~(\ref{newbound11}) 
is presented in Fig.~\ref{R100_plot}. 
\begin{figure}
	\includegraphics[width=\columnwidth]{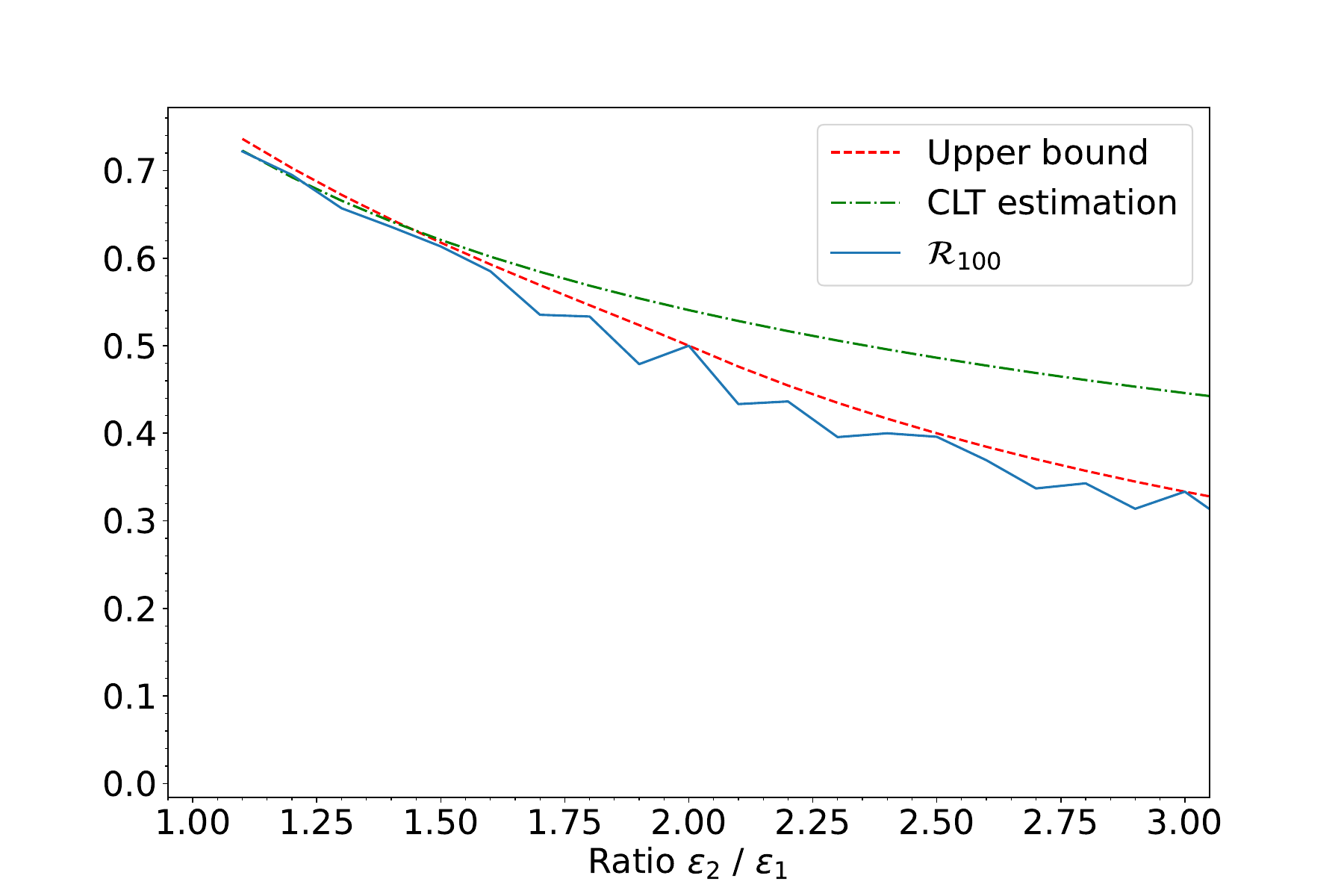}
	\caption{Values of $\mathcal{R}_{n}( \mathcal{A} ; \ham) / \epsilon_{d-1}$ for a 3-level Hamiltonian $\ham = \epsilon_1 \ket{1}\bra{1} + \epsilon_2\ket{2}\bra{2}$, for $\mathcal{A} = \mbox{Span}\{ \ket{1}, \ket{2} \}$ as a function of the (rational) value of the second energy level $\epsilon_2 > 1$ (shown every 0.1). Fixing $n=100$, the finite-size rates $\mathcal{R}_{n}( \mathcal{A} ; \ham)$ are compared with the upper bound~(\ref{newbound11}) and with the heuristic estimation~(\ref{CLT_estimation}) 
	based on the Central Limit Theorem presented in Sec.~\ref{sec:CLT_estimation}. }
	\label{R100_plot}
\end{figure}

\section{Examples}
\label{sec:examples}
In this section we present some simple (yet not-trivial) examples: these configurations serve as an ideal setting to explore the super-additivity effect outlined in Sec.~\ref{sec:superadditivity}, while also facilitating the development of a deeper physical intuition for the problem at hand.

\subsection{Non degenerate 3-level systems} 
 The simplest non-trivial  model we can think of is  
a non-degenerate three-level system Hamiltonian 
\begin{eqnarray}\label{defH111} 
\hat{H} = \sum_{i=0}^2 \epsilon_i |\epsilon_i\rangle\langle \epsilon_i|\;, 
\end{eqnarray} 
with input states $\hat{\rho}\in {\mathfrak{S}}_{\cal{A}}$
which assign non-zero population to just the two top-most energy levels, i.e. 
\begin{eqnarray}\label{spanA} 
{\cal A} = \mbox{Span}\{ |\epsilon_1\rangle, |\epsilon_2\rangle\}={\cal H}_1\oplus {\cal H}_2\;, \quad {\mathbb{S}}=\{ 1,2\}\;, 
\end{eqnarray} 
so that 
\begin{eqnarray} \label{input1} 
\hat{\omega}_{\cal A}(0) = \frac{1}{2}(  |\epsilon_1\rangle\langle \epsilon_1|+ |\epsilon_2\rangle\langle \epsilon_2| )\;. 
\end{eqnarray}
A first example of strong-super-additivity of  the MDEW is obtained by setting 
\begin{eqnarray} \label{examples} 
\begin{cases}
\epsilon_2=3 E\;, \\
\epsilon_1=2 E \;, \\
\epsilon_0=0\;,
\end{cases} 
\end{eqnarray} 
$E >0$ being a fixed constant, see panel {\it a)} of Fig.~\ref{figexamples}. 
It is easy to check that under this condition  the maximum deterministic work we can get from a single
copy of $\hat{\omega}_{\cal A}(0)$   is zero, i.e.
  \begin{eqnarray} {W}_{\max}^{(\rm det)} (\hat{\omega}_{\cal A}(0);\hat{H})=0 \;  
  \Longrightarrow \;  {\cal R}_{1}({\cal A};\hat{H})  =0 \;, 
  \end{eqnarray} 
  (indeed the only value of $W$ which fulfils~(\ref{def243}) for   $i=1$ is $2E$, which however is not acceptable for $i=2$).
Nonetheless it turns out that already for $n=2$ one has 
 \begin{eqnarray} \label{ratio2} 
 {W}_{\max}^{(\rm det)} (\omega^{\otimes 2}_{\cal A}(0);\hat{H}^{(2)}) = 2 E \Longrightarrow 
 R_2({\cal A},\hat{H}) = E\;. 
 \end{eqnarray} 
 This result can be obtained employing a non local unitary $\hat{U}^{(2)}$ that induces the following transitions on the populated energy levels, 
 \begin{eqnarray} 
\begin{cases}
|22\rangle \longrightarrow |11\rangle   \qquad (W=E +E = 2E)\;, \\\\
|21\rangle \longrightarrow |20\rangle    \qquad (W=0+2E=2E)\;,\\
|12\rangle \longrightarrow |02\rangle \qquad (W=2E+0=2E) \;, \\\\
|11\rangle \longrightarrow |10\rangle  \qquad (W=0+2E=2E)\;,\label{solu2bis} 
\end{cases} 
\end{eqnarray} 
where  hereafter we use the shorthand notation $|i j \rangle$ to represents the state   $|\epsilon_i\rangle\otimes |\epsilon_{j}\rangle$.
To see that this  is the optimal solution for $n=2$ notice that according~(\ref{def243}), the only two admissible values of 
$W$ associated with the energy level $|11\rangle$, are $2E$ (attained in Eq.~(\ref{solu2bis})) 
 and $4E$ (reachable e.g. through a unitary that maps $|11\rangle$ into $|00\rangle$). The last possibility however is not
 acceptable since there are no unitary transitions of (say) $|21\rangle$ that could lead to such energy gain (indeed for such level
 the only admissible values of $W$ compatibile with ~(\ref{def243})
 are $E$, $2E$, $3E$, and $5E$).   
 In a similar fashion one can show that
 using the three body unitary $U^{(3)}$ that induces the mapping 
  \begin{eqnarray} 
\begin{cases}
|222\rangle \longrightarrow |201\rangle     \qquad (W=0+3E+E = 4E)\;,  \\\\
|221\rangle \longrightarrow |101\rangle   \qquad (W=E+3E+0 = 4 E)\;, \\
|212\rangle \longrightarrow |011\rangle  \qquad (W=3E+0+E = 4 E)\;,\\
|122\rangle \longrightarrow |110\rangle \qquad (W=0+E+3E = 4 E)\;,\\\\
|211\rangle \longrightarrow |200\rangle  \qquad (W=0+2E+2E = 4 E)\;, \\
|121\rangle \longrightarrow |020\rangle \qquad (W=2E+0+2E = 4 E)\;,\\
|112\rangle \longrightarrow |002\rangle  \qquad (W=2E+2E+0 = 4 E)\;, \\\\
|111\rangle \longrightarrow |100\rangle   \qquad (W=0+2E+2E = 4 E)\;,\\
\end{cases} 
\end{eqnarray} 
we get
 \begin{eqnarray} \nonumber 
 {W}_{\max}^{(\rm det)} (\omega^{\otimes 3}_{\cal A}(0);\hat{H}^{(3)}) = 4 E \Longrightarrow 
 R_3({\cal A},\hat{H}) = \frac{4E}{3}\;,
 \end{eqnarray} 
 which further improves the   MDEW ratio reported in Eq.~(\ref{ratio2}). 
 
\begin{figure}
	\includegraphics[width=\columnwidth]{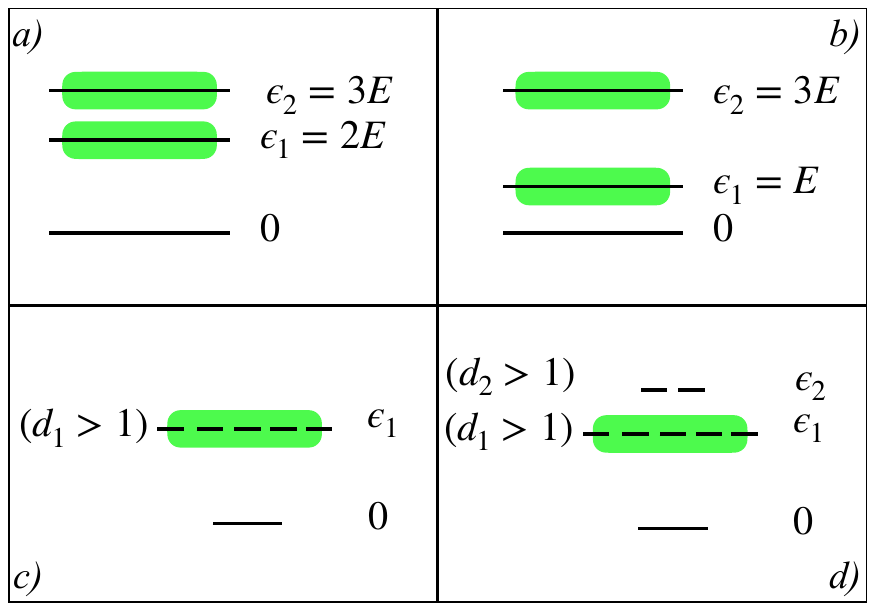}
	\caption{Schematic representation of the examples analyzed in Sec.~\ref{sec:examples}. Panel {\it a)} and {\it b)} describe the  non-degenerate 
	3-level models of Eqs.~(\ref{examples}) and (\ref{ex1}).
	Panel {\it c)} describes  
	2-level model of Eqs.~(\ref{defH111asdfd}) with non trivial degeneracy  associated
	with $\epsilon_1$. Finally panel {\it d)} describe a 3-level model with non-trivial degeneracies for both 
	$\epsilon_1$ and $\epsilon_2$, in which however the highest one is not occupied. 
	In all the examples the green band indicate that the associated level is initially occupied by the input state.}
	\label{figexamples}
\end{figure}
A numerical study of $R_n({\cal A},\hat{H})$ for larger values of $n$ is presented in Fig.~\ref{figuar1}: as evident from the plot in this case, for large $n$ the MDWE ratio approaches the 
upperbound~(\ref{newbound2}).
 \begin{figure}
	\includegraphics[width=\columnwidth]{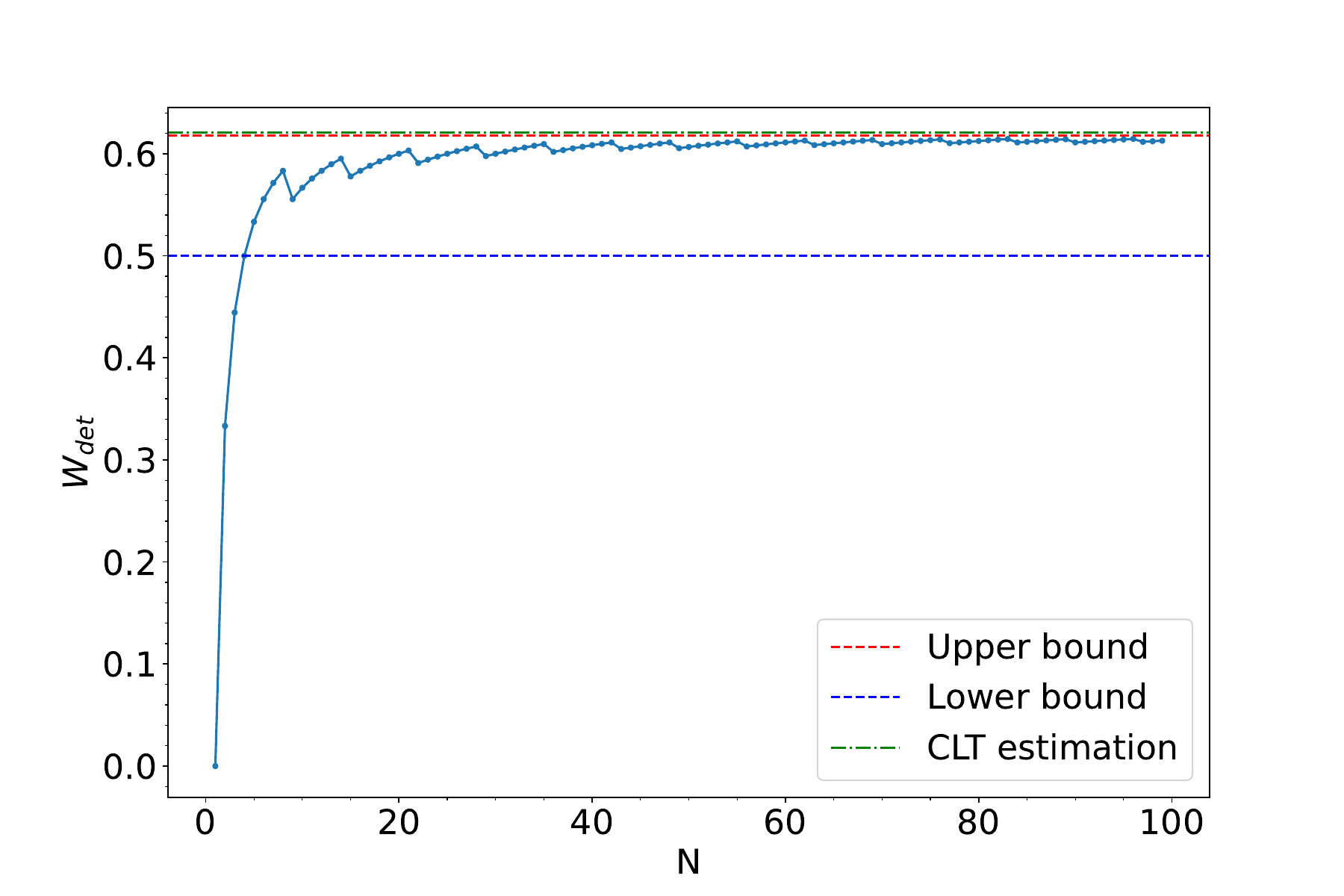}
	\caption{MDEW rate $R_n({\cal A},\hat{H})$ for the Hamiltonian model~(\ref{examples}). The red dashed line corresponds to the upperbound~(\ref{newbound2}) while the blue dashed line to the lower bound~(\ref{lowerlower}) introduced in Sec.~\ref{sec:rational_spectra}.
	 The green mixed line marks the heuristic estimation~(\ref{CLT_estimation}) presented in Sec.~\ref{sec:CLT_estimation}.	 }
	\label{figuar1}
\end{figure}

A class of models (\ref{defH111}) for which the asymptotic ratio $R({\cal A},\hat{H})$ can be explicitly computed is obtained by setting
\begin{eqnarray} 
\begin{cases}
\epsilon_2=3E\;, \\
\epsilon_1=E \;, \\
\epsilon_0=0\;,
\end{cases} \label{ex1}
\end{eqnarray}
see panel {\it b)} of Fig.~\ref{figexamples}.
Notice that in this case the energy gap $\epsilon_1-\epsilon_0=E$ is {\it half} of the energy gap $\epsilon_2-\epsilon_1=2 E$, while in the previous example it was exactly the opposite. 
Notice also that for ${\cal A}$ as in Eq.~(\ref{spanA}) 
 the upper bound Eq.~(\ref{uppernew}) implies 
 \begin{eqnarray}{\cal R}({\cal A}^{\otimes n};\hat{H}^{(n)}) \leq {\cal R}({\cal A};\hat{H}) \leq E \;. \label{uppernew1}
 \end{eqnarray} 
It is easy to see that similarly to the  model of Eq.~(\ref{input1}),
 also in this case we have 
\begin{eqnarray} 
{W}_{\max}^{(\rm det)} (\hat{\omega}_{\cal A}(0);\hat{H})=0\;.  
\end{eqnarray} 
For $n=2$ we get the same result, i.e. 
\begin{eqnarray} 
{{W}_{\max}^{(\rm det)} (\hat{\omega}^{\otimes 2}_{\cal A}(0);\hat{H}^{(2)})}=0\;. 
\end{eqnarray} 
To see this notice that  here from $|11\rangle$ we have only two possible transitions: toward $|10\rangle$ or $|01\rangle$ which corresponds to the extraction of an energy $E$, or
toward $|00\rangle$ with extraction of energy $2E$. From the doublet $|12\rangle$, $|21\rangle$ we have only two possibility:
either a transition toward $|02\rangle,|20\rangle$ with energy $E$, or a transition toward  $|01\rangle,|10\rangle$  with energy extraction of $3E$.
So we have a match for $E$. However there are no transitions for $|22\rangle$ that produces such an amount of the energy (the minimum energy we can extract from such level is indeed $2E$). 
What about $n=3$? In this case we observe that the bound (\ref{uppernew1}) gets saturated, i.e.
 \begin{equation}{W_{\max}^{(\rm det)}(\hat{\omega}^{\otimes 3}_{\cal A}(0);\hat{H}^{(3)})
  } =3E  \Longrightarrow 
  R_3({\cal A}; \hat{H})  =E\;,   \label{sat}
 \end{equation} 
 implying that $E$ is the asymptotic MDEW ratio of the model. 
 The result of Eq.~(\ref{sat}) 
 is achieved with the choice of the unitary $\hat{U}^{(3)}$ which induce the mappings
 \begin{eqnarray} 
\begin{cases}
|222\rangle \longrightarrow |220\rangle    \qquad (W=0+0+3E = 3E)\;,  \\\\
|221\rangle \longrightarrow |021\rangle   \qquad (W=3E+ 0+0= 3E)\;,  \\
|212\rangle \longrightarrow |210\rangle  \qquad (W= 0+0+3E = 3E)\;, \\
|122\rangle \longrightarrow |102\rangle   \qquad (W=0+3E +0= 3E)\;, \\\\
|211\rangle \longrightarrow |011\rangle  \qquad (W=3E +0+0= 3E)\;,\\
|121\rangle \longrightarrow |101\rangle \qquad (W=0+3E +0= 3E)\;,\\
|112\rangle \longrightarrow |110\rangle  \qquad (W=0+0+3E= 3E)\;, \\\\
|111\rangle \longrightarrow |000\rangle  \qquad  (W=E+E+E= 3E)\;.\\
\end{cases} 
\end{eqnarray} 
Equation~(\ref{sat})  can be extended to the whole class of energy spectra of the form 
\begin{eqnarray} 
\begin{cases}
\epsilon_2=N E\;, \\
\epsilon_1= E \;, \\
\epsilon_0=0\;,
\end{cases} \label{ex1bis}
\end{eqnarray} 
with $N\geq 2$ integer. 
To see this take $n=N$ and use the unitary $\hat{U}^{(N)}$ which induces the mapping
 \begin{eqnarray} 
\begin{cases}\nonumber 
|22\cdots 22\rangle \longrightarrow |02\cdots 22\rangle    \quad  \mbox{\small$(W=NE + 0+\cdots +0 = NE)$,} \\\\
|22\cdots 21\rangle \longrightarrow |02\cdots 21\rangle     \quad  \mbox{\small$(W=NE + 0+\cdots +0 = NE)$,} \\
\qquad \mbox{{\it permutations}}    \\\\
|22\cdots 211\rangle \longrightarrow |02\cdots 211\rangle \;\; \mbox{\small$(W=NE + 0+\cdots +0 = NE)$,} \\
\qquad \mbox{{\it permutations}}     \\\\
\qquad \qquad  \vdots 
\\\\
|11\cdots 11\rangle \longrightarrow |00\cdots 00\rangle    \quad  \mbox{\small$(W=E + E+\cdots +E = NE)$,}\\
\end{cases} 
\end{eqnarray} 
(in other words $\hat{U}^{(N)}$ maps $|11\cdots 11\rangle$ into the ground while when acting any other eigenvector of $\hat{\omega}_{\cal A}^{\otimes N}(0)$ replace one (and only one) of the $2$ terms with a $0$).
By construction we have that 
\begin{equation}{W_{\max}^{(\rm det)}(\hat{\omega}^{\otimes N}_{\cal A}(0);\hat{H}^{(N)})
  } =NE  \Longrightarrow 
  R_N({\cal A}; \hat{H})  =E\;,   \label{satN}
 \end{equation} 
which once more saturates the bound (\ref{uppernew}).

\subsection{Degenerate 2-level models} \label{sec:deg2lm} 
Adding degeneracy in the model typically increases the complexity of the MDEW analysis. 
Consider for instance a two-level model  with degeneracy $d_1 >1$ for  the excited level $\epsilon_1$ and with  the ground level 
 $\epsilon_0=0$ that has no degeneracy, i.e. 
 \begin{eqnarray}\label{defH111asdfd} 
\hat{H} = \epsilon_1  \sum_{j=0}^{d_1-1} |\epsilon_{1,j} \rangle\langle \epsilon_{1,j}| \;, \end{eqnarray} 
 with input states  $\hat{\rho}\in {\mathfrak{S}}_{\cal{A}}$
 which assign non-zero population to
 all the elements of the excited level leaving the ground level empty, i.e. 
\begin{eqnarray}\label{spanAAAA} 
{\cal A} = \mbox{Span}\{ |\epsilon_{1,0} \rangle, 
\cdots , |\epsilon_{1,d_1-1}\rangle\}={\cal H}_1\;, \quad {\mathbb{S}}=\{ 1\}\;,
\end{eqnarray} 
see panel {\it c)} of Fig.~\ref{figexamples}.
 Clearly for $n=1$ we have that no energy can be extracted in the absence of fluctuations, i.e.
  \begin{eqnarray}\label{rndsadfadf} {W}_{\max}^{(\rm det)} (\hat{\omega}_{\cal A}(0);\hat{H})=0 \;  
  \Longrightarrow \;  {\cal R}_{1}({\cal A};\hat{H})  =0 \;.  \end{eqnarray} 
  The situation changes however already for $n=2$. 
 Indeed in this case, if $d_1=2$ we can get a rate of $\epsilon_1/2$,
  by using the following unitary operation:
  \begin{eqnarray} 
\begin{cases}
|1_01_0\rangle \longrightarrow |01_0\rangle  \;,  \\
|1_11_1\rangle \longrightarrow |01_1\rangle    \;,\\
|1_01_1\rangle \longrightarrow |1_00\rangle \;, \\
|1_11_0\rangle \longrightarrow |1_10\rangle \;,\label{solu2} 
\end{cases} 
\end{eqnarray} 
where we used $|1\rangle$, $|1'\rangle$ to represent the two orthogonal states
$|\epsilon_{1,0}\rangle$ and $|\epsilon_{1,1}\rangle$ 
 of level 1. 
More generally, for $d_1>1$ generic, we can use $n$ copies of the state we could extract the work $k \epsilon_1$  by promoting $k$ excited states into the ground level if the following conditions
are satisfied
\begin{eqnarray} 
\#\mbox{input}(n)  \leq \#\mbox{output}(n,k)\;, \label{ff} 
\end{eqnarray} 
 where $\#\mbox{input}(n) $ and $\#\mbox{output}(n,k)$ are the number of orthogonal configurations associated, rispectively, with the $n$ copies of the input state  and they transformed versions. 
 The first number corresponds to the possible ways in which we form $n$-long strings with $d_1$ symbols, i.e. $\#\mbox{input}(n)  = d_1^n$,
while the second corresponds to the possible ways in which we can form $n$-long strings using $d_1$ symbols under the constraint that $k$ elements are fixed equal to zero, i.e.
$\#\mbox{output}(n,k) =\left(\begin{array}{c}n\\ k\end{array} \right)  d_1^{n-k}$. 
Accordingly Eq.~(\ref{ff}) reduces to the constraint
\begin{eqnarray} 
C_{n,k}(d_1) := \left(\begin{array}{c}n\\ k\end{array} \right)  - d_1^k\geq 0\;. \label{constasdadf} 
\end{eqnarray} 
 For each $n$, and $k$ fulfilling the above expression $k/n$ represents an achievable rate. Observe also that for each fixed $n$ the maximum $k$ that is compatible 
 with~(\ref{constasdadf}) represents the maximum rate attainable (indeed the only way we have to get energy from the system is to promote the excited state into the ground).
 Accordingly we can write 
  \begin{eqnarray}\label{rn} 
 {\cal R}_{n}({\cal A};\hat{H})  &=& \epsilon_1 \max_{k}  \{ k / n :  C_{n,k}(d_1) \geq 0\}\;, \\ 
  {\cal R}({\cal A};\hat{H})  &=&  \epsilon_1  \lim_{n\rightarrow \infty} \max_{k}  \{ k / n :  C_{n,k}(d_1) \geq 0\} \nonumber \\ 
  &=&  \epsilon_1 \max_{n,k}  \{ k / n :  C_{n,k}(d_1) \geq 0\} \;. \label{rinf} 
 \end{eqnarray} 
 Recalling that for all $n$ and $k$ we have 
 \begin{eqnarray} \label{upper} 
\left(\frac{e n}{k} \right)^k \geq \left(\begin{array}{c}n\\ k\end{array} \right)  \geq \left(\frac{n}{k} \right)^k \;. \label{const1} 
\end{eqnarray} 
The lower bound implies that   for $n,k$ such that  $\frac{k}{n}  \leq \frac{1}{d_1}$ one has  $C_{n,k}(d_1)\geq 0$; the upper bound instead can be used to 
verify that $\frac{k}{n}  >  \frac{e}{d_1}$  instead we always get $C_{n,k}(d_1) < 0$.
 Replacing this in the above expression yields the following bounds for ${\cal R}({\cal A};\hat{H})$, 
\begin{eqnarray}   e\epsilon_1/d_1 \geq  {\cal R}({\cal A};\hat{H})  \geq  \epsilon_1/d_1\;.
\end{eqnarray} 
The above analysis  can be easily extended to include also those configurations where the
input states of the system do occupy all the full energy subspace ${\cal H}_1$ associated with the energy level $\epsilon_1$. In fact suppose that ${\cal A}$ covers only $\delta_1<d_1$ of the vectors of ${\cal H}_1$, e.g.
\begin{eqnarray}\label{spanAAAA111} 
{\cal A} = \mbox{Span}\{ |\epsilon_{1,0} \rangle, 
\cdots , |\epsilon_{1,\delta_1-1}\rangle\}\subset{\cal H}_1\;, \quad {\mathbb{S}}=\{ 1\}\;.
\end{eqnarray} 
Under this condition we can still use Eq.~(\ref{ff}) to identify the  work values which can be extracted deterministically: in this case however the left-hand-side term of such inequality assumes a smaller value (i.e. $\#\mbox{input}(n)=\delta_1^n$), and (\ref{constasdadf}) gets replaced by the 
weaker constraint 
\begin{eqnarray} 
C_{n,k}(d_1,\delta_1) := \left(\begin{array}{c}n\\ k\end{array} \right)  
- d_1^k \left(\frac{\delta_1}{d_1}\right)^{n} \geq 0\;.
\label{const12asdfadf} 
\end{eqnarray} 
Inserting this into~(\ref{rn}) and (\ref{rinf}) leads to MDEW rates which are larger than or equal to the one obtained for $\delta_1=d_1$ in agreement with the prediction of Lemma~\ref{lemma33}. 
For instance for $\delta_1=1$ and $d_1\geq 2$ the inequality~(\ref{const12asdfadf}) can be always fulfilled with $k=n$ leading to  ${\cal R}_n({\cal A};\hat{H})=\epsilon_1$ which corresponds to  the maximum work one can hope to extract from the system. 
More generally the new values of the rates are given by  
  \begin{eqnarray}\label{rnnasdaf} 
 {\cal R}_{n}({\cal A};\hat{H})  &=& \epsilon_1  \max_{k}  \{ k / n :  C_{n,k}(d_1,\delta_1)  \geq 0\}\;, \\ 
  {\cal R}({\cal A};\hat{H})  &=&  \epsilon_1 \lim_{n\rightarrow \infty} \max_{k}  \{ k / n :  C_{n,k}(d_1,\delta_1)  \geq 0\} \nonumber \\ 
  &=& \epsilon_1  \max_{n,k}  \{ k / n :  C_{n,k}(d_1,\delta_1)  \geq 0\} \;. \label{rinf2} 
 \end{eqnarray}

\subsection{Free levels at higher energy} 
As established by Lemma~\ref{lemma33} reducing the occupancies numbers 
$\overline{d}_i = \mbox{dim}[{\cal A}_i]$ 
of the energy eigenspaces of the model tends to improve the MDEW of the model: this is a direct consequence of the fact that smaller values of the $\overline{d}_i$'s corresponds to weaker constraints on the associated optimization problem.  
A similar effect arises when we increase the degeneracy of $\hat{H}$ while keeping the same 
occupation level of ${\cal A}$. For instance in the example of Sec.~\ref{sec:deg2lm}, setting
$d_0=d_1$, for ${\cal A}$ as in Eq.~(\ref{spanAAAA}) will always allow for an optimal MDEW value rate of ${\cal R}_n({\cal A};\hat{H})=\epsilon_1$. Strangely enough the same phenomenon can also occur if we add  extra levels to $\hat{H}$ with energy values that are {\it above} the one occupied by ${\cal A}$. 
To see this consider the case of a 3-level model Hamiltonian $\hat{H}$ obtained by adding  an extra level $\epsilon_2> \epsilon_1$ with degeneracy $d_2\geq 1$ to the one presented
in~(\ref{defH111asdfd}), i.e. 
 \begin{eqnarray}\label{defH111asdfd1} 
\hat{H} = \epsilon_1  \sum_{j=0}^{d_1-1} |\epsilon_{1,j} \rangle\langle \epsilon_{1,j}| +
\epsilon_2 \sum_{j=0}^{d_2-1} |\epsilon_{2,j} \rangle\langle \epsilon_{2,j}|\;. \end{eqnarray} 
while maintaining ${\cal A}$ as in (\ref{spanAAAA}), i.e. assigning zero occupation to both
the ground level and the new one, and assuming full occupancy for the intermediate level --
see panel {\it d)} of Fig.~\ref{figexamples}.
 Under these assumptions one would be tempted to conclude that the level $\epsilon_2$ plays no fundamental role in the energy extraction process: indeed
promoting populations from $\epsilon_1$ to $\epsilon_2$ will cost an energy $\epsilon_2-
\epsilon_1$ which will contribute negatively on the overall budget. It turns out however that
under certain conditions such a loss can be exploited to improve the MDEW efficiency above the one described in Eq.~(\ref{rn}) -- which in the  context of the 3-level model corresponds to the restricted set of TPM strategies  where we can only move 
population from $\epsilon_1$ toward the ground state. 
To see this consider for instance the case where we have at disposal $n$ copies of the input
state $\hat{\rho}\in {\mathfrak{S}}_{\cal{A}}$.
Given hence $k^{\star}$ the maximum $k$ that fulfils~(\ref{constasdadf}), from (\ref{rn}) we know that 
the strategies that convert states of $\epsilon_1$ into the ground can
 achieve at most the rate 
  \begin{eqnarray}\label{rnopt} 
\tilde{\cal R}_{n}({\cal A};\hat{H}) = \epsilon_1 k^{\star}/n \;, 
 \end{eqnarray} 
  (notice that typically this will be smaller than $\epsilon_1$ since
 $k^{\star}< n$). 
  Exploiting the presence of $\epsilon_2$ we can try to do better e.g. promoting $k^{\star}+1$ states 
 $\epsilon_1$  into the ground {\it and} one extra state  $\epsilon_1$ into one of the levels $\epsilon_2$.
Indeed assuming that such unitary exists we could gain a rate equal to
\begin{eqnarray} \label{value} 
 {\cal R}_{n} &=& \frac{\epsilon_1(k^{\star}+1) -( \epsilon_2-\epsilon_1)}{n} =
 \frac{\epsilon_1k^{\star} +  (2\epsilon_1 -\epsilon_2)}{n} \nonumber \\
 &=& \tilde{\cal R}_{n}({\cal A};\hat{H}) + \frac{(2\epsilon_1 -\epsilon_2)}{n} \;, 
\end{eqnarray} 
which   is greater than $\tilde{\cal R}_{n}({\cal A};\hat{H})$ whenever $2\epsilon_1 >\epsilon_2 $.
A sufficient condition for this to happens, is that there are sufficiently many output configurations with 
$k^{\star}+1$ ground states, $n-(k^{\star}+2)$ states $\epsilon_1$ and one state $\epsilon_2$, to accomodate
the input configurations $\#\mbox{input}(n)  = d_1^n$ of $\hat{\rho}^{\otimes n}$.
Considering the degeneracy we have assumed for $\hat{H}$, the total number of the above output configurations can be explicitly computed: they are 
 \begin{equation} 
\#\mbox{output}(n,k^{\star}+1,1) :=\frac{n! \; d_1^{n-(k^{\star}+1)} d_2 }{ (n-k^{\star}-2)!(k^{\star}+1)!1!}   \;.
\end{equation}
Accordingly the possibility of reaching the rate~(\ref{value}) is determined by the inequality
\begin{eqnarray} \label{dd2} 
 d_2 \geq \left( \frac{(n-k^{\star}-2)!(k^{\star}+1)!}{n!} \right)  d_1^{k^{\star}+1}\;.
\end{eqnarray}
As an example consider for instance what happens for $d_1=2$ and $n=3$. 
Under this condition one notices that $k^{\star}=1$, so that $k^{\star}/n=1/3$. On the contrary
the condition~(\ref{dd2}) becomes $d_2 \geq 4/3$: therefore it is sufficient to have $d_2=2$ to bring the rate from 
$\epsilon_1/3$ to
$(\epsilon_1 + (2\epsilon_1-\epsilon_2))/3$.
Notice that the presence of $d_2$ can be exploited to lead even more drastic improvements: for instance, for fixed $n$, one can try to promote $n-1$ states to the ground paying the price of having a single state in $\epsilon_2$.
Under this condition one could push the rate at
\begin{eqnarray} \label{valueR} 
 {\cal R}_{n} = \frac{(n-1)\epsilon_1-(\epsilon_2-\epsilon_1) }{n}= \epsilon_1- \frac{ \epsilon_2}{n} \;, 
\end{eqnarray} 
 which for $n$ sufficiently large approximates the upper bound 
 $\epsilon_1$ dictated by~(\ref{uppernew}).
 The condition for this to happens is that $d_2$ is sufficiently large to ensure that the output configurations
 with $n-1$ ground states and $1$ state $\epsilon_2$ are larger than  $\#\mbox{input}(n)  = d_1^n$, i.e.
 \begin{eqnarray} 
 n d_2 \geq d_1^n\;, \qquad \Longrightarrow \qquad d_2 \geq d_1^n/n\;. 
 \end{eqnarray}

\section{Rational spectra} 
\label{sec:rational_spectra}
Building up from the examples analyzed in the previous section, here 
we focus on a special class of models for which one can explicitly prove that
the asymptotic MDEW ratio  is 
non zero.
Specifically we shall consider the case where   
the non-empty elements set ${\mathbb{S}}$ of the subspace ${\cal A}$ identifies energy levels 
of the Hamiltonian $\hat{H}$ that are proportional to integer numbers  up to a common multiplicative factor~$E$:
\begin{equation} \label{imponew} 
\forall i\in \mathbb{S}\quad  \begin{cases}  \hat{\Pi}_{{\cal A}_i} = \sum_{s=0}^{\overline{d}_i-1} |\epsilon_{i,s}\rangle\langle \epsilon_{i,s}| \;, \;\; 
\hat{H} |\epsilon_{i,s}\rangle = \epsilon_{i} |\epsilon_{i,s}\rangle\;,  \\ \\ 
\forall i,s,i^\prime,s^\prime : \braket{\epsilon_{i,s} | \epsilon_{i^\prime,s^\prime}} = 
\delta_{ii^\prime} \delta_{ss^\prime} \;, \\ \\
 \exists  \;  m_i \in \mathbb{N} : \epsilon_i= E \;  m_i \;,
 \end{cases} 
\end{equation} 
where 
\begin{eqnarray} \overline{d}_i:= \mbox{dim}[ {\cal A}_i]\leq d_i\;,
\end{eqnarray}
 is the dimension of the $i$-th energy block ${\cal A}_i$ of ${\cal A}$, and 
$\{ |\epsilon_{i,s}\rangle\}_{s=1,\cdots, \overline{d}_i}$ an orthonormal basis for such space  
(of course such scenario includes as special instances the settings where the entire spectrum
of $\hat{H}$ -- not just the part of it that it is filtered out in $\cal A$ -- fulfils  the above requirement). 
Under the condition~(\ref{imponew}) we can prove that, as long as the ground state of the system is not populated, i.e. if 
$0\notin \mathbb{S}$, the asymptotic MDEW ratio of the model is 
explicitly non zero (of course if $0\in \mathbb{S}$ then 
the MDEW ratio is always null due to Eq.~(\ref{nogo})).
 In order to do so we shall provide a lower bound for ${\cal R}({\cal A};\hat{H})$ which is explicitly not zero. 

Assume hence  $\mathbb{S}$ to be a collection of energy eigenvectors  indexes which does not include the ground energy level and that contains at least two distinct elements (the case in which $\mathbb{S}$ has a unique element is already solved in Eq.~(\ref{ide1})). 
Define then $M_\mathbb{S}$  to be the least common multiplier of the integers $m_i$'s associated with the populated  part of the spectrum of $\hat{H}$, i.e. 
\begin{equation} 
M_\mathbb{S} := \mbox{lcm}\left\{ m_i : i \in  \mathbb{S}\right\}\;. 
\label{def_Ms}
\end{equation} 
and the quantities 
\begin{eqnarray} 
K_i &:=& \frac{M_\mathbb{S}}{m_i} + \overline{d}_i - 1  \;, \label{def_Ki} \\
K_\mathbb{S}&:=& \sum_{i\in  \mathbb{S} } \overline{d}_i (K_i-1) \label{def_Ks} \;.
\end{eqnarray} 
Notice that since by construction  the $m_i$'s are all distinct integer numbers greater than  or equal to 1, we
can conclude that $M_\mathbb{S} \geq 2$. Also it follows that  the $K_i$'s are all greater than or equal to 1 and that  
\begin{eqnarray} \label{mmm} 
\frac{M_\mathbb{S}}{m_i}  \neq \frac{M_\mathbb{S}}{m_j} \;,  \qquad \forall i\neq j  \in  \mathbb{S} \;. 
\end{eqnarray} 
Observe next that from Eq.~(\ref{imponew}) it follows that 
 for $n$ integer,  the eigenvectors of 
 $\hat{\omega}^{\otimes n}_{\cal A}(0)$ are provided by the tensor product states of the form \begin{eqnarray}
|\epsilon_{\vec{i},\vec{s}}\rangle
 : = | \epsilon_{{i}_1, s_1}\rangle \otimes | \epsilon_{{i}_2, s_2}\rangle \otimes  \cdots \otimes | \epsilon_{{i}_n, s_n}\rangle \;, 
 \end{eqnarray} 
 meaning that each eigenstate can be uniquely identify by a couple $(\vec{i}, \vec{s})$,
 with $\vec{i}:= (i_1,i_2,\cdots, i_n)\in \mathbb{S}^n$ and $\vec{s}:= (s_1,s_2,\cdots, s_n)$. Let $\mathbb{V}$ denote the set of allowed vectors:
 \begin{eqnarray}
 \mathbb{V} := \left\{ ( \vec{i}, \vec{s} ) : \vec{i} \in  \mathbb{S}^{n}, \forall j \;  0 \leq s_j < \overline{d}_i \right\}\;.
 \end{eqnarray}
 For each couple $(\vec{i}, \vec{s}) \in  \mathbb{V}$, we define $n_{j,u}(\vec{i}, \vec{s})$ as the number of copies of
 the terms $|\epsilon_{j,u}\rangle$ it contains: these quantities of course provide a partition of $n$, i.e.
 $\sum_{j\in \mathbb{S}} \sum_{u=0}^{\delta_j-1} n_{j,u}(\vec{i}, \vec{s})= n$.
 Observe also that 
 \begin{lemma}\label{lemma220}
 Given $n> K_\mathbb{S}$, for each $(\vec{i}, \vec{s})\in \mathbb{V}$ 
  there exists $j^{\star}\in  \mathbb{S}$ and $u^\star \in \{0,1,\cdots, \delta_{j^\star} - 1\}$ such that  $n_{j^{\star}, u^{\star}}(\vec{i}, \vec{s}) \geq K_{j^{\star}}$. 
\end{lemma}
\begin{proof} Assume  by contraddiction that 
all the $n_{j,u}(\vec{i}, \vec{s})$'s are smaller than the corresponding $K_j$'s. Then we can write 
\begin{eqnarray} 
n = \sum_{j\in \mathbb{S}} \sum_{u=0}^{\delta_j - 1} n_{j,u}(\vec{i}, \vec{s})&\leq& \sum_{j\in \mathbb{S}}  \sum_{u=0}^{\delta_j - 1} (K_j -1) \nonumber \\
&=& \sum_{j\in \mathbb{S}} \delta_j (K_j -1) = K_\mathbb{S} \;,
\end{eqnarray} 
which is impossibile. \end{proof}
As a consequence of the above result it follows that as long as $n\geq K_\mathbb{S}+1$, then
for each $( \vec{i}, \vec{s} ) \in \mathbb{V}$ we can  assign the quantities
\begin{eqnarray}\begin{cases}
{j}_{\star}(\vec{i}) &\;:= \;\;\min \left\{ j_\star\in \mathbb{S} : \exists u \mbox{ s.t. }n_{j_\star, u}(\vec{i}, \vec{s})\geq K_{j_\star}\right\}\;,\\
{u}_{\star}(\vec{i}, \vec{s}) &\;:=\;\;\min \left\{ u_\star : n_{j_\star(\vec{i}), u_\star}(\vec{i}, \vec{s})\geq K_{j_\star}\right\}\;,\\
K_{\star}(\vec{i}) &\;:= \;\; K_{j={j}_{\star}(\vec{i})}\;.\end{cases} \nonumber 
\end{eqnarray} 
The set  $\mathbb{V}$ can  hence be divided into a collection of disjoint subsets 
which contain vectors $( \vec{i}, \vec{s} )$ that have the same values of ${j}_{\star}(\vec{i})$ and ${u}_{\star}(\vec{i}, \vec{s})$ (and hence the same 
$K_{\star}(\vec{i})$), i.e. 
\begin{eqnarray} 
\begin{cases} \mathbb{V} \;&:= \;\; {\bigcup}_{a\in \mathbb{S}} {\bigcup}_{0 \leq b < \delta_a} \mathbb{V}_{a,b}\;, \\
\mathbb{V}_{a,b} \;&:= \;\;  \{ ( \vec{i}, \vec{s} ) \in \mathbb{V}: {j}_{\star}(\vec{i})  = a , \; {u}_{\star}(\vec{i}, \vec{s})  = b \}\;.
\end{cases} \end{eqnarray} 
By construction the couples of vectors included in $\mathbb{V}_{a,b}$ possess at least $K_a=
 \frac{M_\mathbb{S}}{m_a} + \delta_a - 1$ copies of 
the symbol ${(a,b)}$. 
 For each $\mathbb{V}_{a,b}$ we can hence assign a new set of couples $n$-dimensional vectors $\overline{\mathbb{V}}_{a,b}$ 
whose elements are obtained by taking the vectors of 
${\mathbb{V}}_{a,b}$, and replacing $M_\mathbb{S} / m_a$ copies of the entry $(a,b)$ with $(0,0)$. Since for each $( \vec{i}, \vec{s} ) \in {\mathbb{V}}_{a,b}$ there are at least ${K_a \choose M_\mathbb{S} / m_a} = {M_\mathbb{S} / m_a + \delta_a -1 \choose M_\mathbb{S} / m_a}$ ways to do this, the size of $\overline{\mathbb{V}}_{a,b}$
must satisfy
\begin{eqnarray}
|\overline{\mathbb{V}}_{a,b}| \geq {M_\mathbb{S} / m_a + \delta_a -1 \choose M_\mathbb{S} / m_a}  |{\mathbb{V}}_{a,b}|  \geq \delta_a |{\mathbb{V}}_{a,b}| \;,
\end{eqnarray}
meaning that for each ${\mathbb{V}}_{a,b}$ we can identify a subset $\tilde{\mathbb{V}}_{a,b} \subseteq \overline{\mathbb{V}}_{a,b}$ whose cardinality is exactly $|\tilde{\mathbb{V}}_{a,b}| = \delta_a |{\mathbb{V}}_{a,b}|$.
Now let $\mathbb{V}_{a} := \bigcup_{b=0}^{\delta_a-1} {\mathbb{V}}_{a,b}$ and $\overline{\mathbb{V}}_{a} := \bigcup_{b=0}^{\delta_a-1} \overline{\mathbb{V}}_{a,b}$. The sets ${\mathbb{V}}_{a,b}$ are by construction disjoint, and therefore 
\begin{eqnarray}
|\mathbb{V}_{a}| = \sum_{b=0}^{\delta_a-1} |{\mathbb{V}}_{a,b}| \; . 
\end{eqnarray}
The size of the set $\overline{\mathbb{V}}_{a}$ satisfies instead the inequality 
\begin{eqnarray}
 |\overline{\mathbb{V}}_{a}| &=&  \left| \bigcup_{b=0}^{\delta_a-1} \overline{\mathbb{V}}_{a,b}\right| \geq \left| \bigcup_{b=0}^{\delta_a-1} \tilde{\mathbb{V}}_{a,b}\right| 
  \geq \max_b |\tilde{\mathbb{V}}_{a,b}| \nonumber \\ &=& \delta_a  \max_b |{\mathbb{V}}_{a,b}|   \geq \sum_{b=0}^{\delta_a-1} |{\mathbb{V}}_{a,b}| =  |\mathbb{V}_{a}|  \; .
\end{eqnarray}

Recall that the elements of $\overline{\mathbb{V}}_{a,b}$ are characterized by $M_\mathbb{S} / K_a$ copies of the ground state. From Eq.~(\ref{mmm}) it follows that the sets $\overline{\mathbb{V}}_{a}$ do not overlap, i.e.
\begin{eqnarray} 
  \overline{\mathbb{V}}_{a} \cap\overline{\mathbb{V}}_{a'}= \varnothing\;, \quad \forall a \neq a' \in \mathbb{S} \;.
  \end{eqnarray} 
  Accordingly we can  identify a  mapping ${\cal F}$ from 
  ${\mathbb{V}}={\bigcup}_{a\in \mathbb{S}} {\mathbb{V}}_a$ to ${\bigcup}_{a\in \mathbb{S}}  \overline{\mathbb{V}}_a$
which for all $a$ sends ${\mathbb{V}}_a$ into a  subset of $\overline{\mathbb{V}}_a$,
\begin{eqnarray} \label{mapping1} 
(\vec{i},\vec{s}) \in {\mathbb{V}}_a \mapsto {\cal F}(\vec{i},\vec{s})\in \overline{\mathbb{V}}_a \;,
\end{eqnarray} 
which is injective, i.e. such that $ {\cal F}(\vec{i},\vec{s})\neq  {\cal F}(\vec{i}',\vec{s}')$ for all
$(\vec{i},\vec{s}) \neq (\vec{i}',\vec{s}')$.
From this we can now derive the following lower bound for the MDEW ratio, 
\begin{eqnarray}\label{lowerlower} 
{\cal R}({\cal A};\hat{H})\geq  E \frac{M_{\mathbb{S}}}{K_{\mathbb{S}}+1}>0\;. 
\end{eqnarray} 
The proof relays on the observation that for $n \geq  K_{\mathbb{S}}+1$ there exists a unitary transformation $\hat{U}_{\cal F}^{(n)}$ which 
enables us to extract an amount $W=EM_{\mathbb{S}}$ of work deterministically. On the eigenvectors $|\epsilon_{\vec{i},\vec{s}}\rangle$
which form the support of $\hat{\omega}^{\otimes n}_{\cal A}(0)$ such unitary is simply the transformation which implement 
the mappings~(\ref{mapping1}) defined above, i.e. 
\begin{eqnarray} 
 \hat{U}_{\cal F}^{(n)} |{\epsilon}_{\vec{i},\vec{s}} \rangle =  |\vec{\epsilon}_{{\cal F}({\vec{i},\vec{s}})}\rangle\;, \qquad \forall(\vec{i},\vec{s})  \in \mathbb{V} \;.
\end{eqnarray} 
Due to the fact that  for $(\vec{i},\vec{s})\in\mathbb{V}_a$, the states
$|{\epsilon}_{(\vec{i},\vec{s})} \rangle$ and $|\vec{\epsilon}_{{\cal F}(\vec{i},\vec{s})}\rangle$ only differs by the fact that in the latter 
$K_a$ copies of eigenvectors with energy eigenvalue $\epsilon_a$ are replaced with a ground state vector, the associated energy gain for each one for all these transitions  is equal to
\begin{eqnarray} 
W = K_a \epsilon_a =  \frac{M_\mathbb{S}}{m_a} E m_a= E M_\mathbb{S}\;.
\end{eqnarray} 
Accordingly for $\hat{\rho}$ with support ${\cal H}_{\mathbb{S}}$ we can write 
\begin{eqnarray} \label{orbd} 
&&\left\{ \begin{array}{l} \langle W_ {\hat{U}_{\cal F}^{(n)}}(\hat{\omega}^{\otimes n}_{\cal A}(0); \hat{H}^{(n)}) \rangle = E M_\mathbb{S}\;,\\
\langle \Delta^2W_ {\hat{U}_{\cal F}^{(n)}}(\hat{\omega}^{\otimes n}_{\cal A}(0); \hat{H}^{(n)}) \rangle = 0\;,\end{array} \right.  \\ \nonumber
&&\qquad \qquad \qquad \qquad \qquad \Longrightarrow 
{\cal R}_n({\cal A};\hat{H})\geq \frac{E M_\mathbb{S} }{n}\;,
\end{eqnarray} 
The maximum of the above expression  is achieved for $n=K_{\mathbb{S}}+1$ which via~(\ref{impo11}) finally leads to (\ref{lowerlower}). 

In the special case in which the spectrum of the Hamiltonian is non-degenerate (i.e.  $d_i=\overline{d}_i=1 \implies K_i = M_{\mathbb{S}}/m_i$ for all $i$), using in~(\ref{lowerlower}) the definitions~(\ref{def_Ms}),~(\ref{def_Ki}) and~(\ref{def_Ks}), we can recast the lower bound~(\ref{lowerlower}) in a 
slightly weaker form which unveils a more straightforward and useful dependence on the energy levels $\epsilon_i$, i.e. 
\begin{eqnarray}
\label{harmonic_mean}
{\cal R}({\cal A};\hat{H})\geq E \left( \sum_{i \in \mathbb{S} } \frac{1}{m_i} \right)^{-1}   =\left( \sum_{i \in \mathbb{S}} \frac{1}{\epsilon_i} \right)^{-1}\;.
\end{eqnarray}
The difference between the bounds~(\ref{harmonic_mean}) and~(\ref{lowerlower}) becomes negligible when all the $m_i$ satisfy $m_i \gg d$.

\subsection{Finite size behaviour}
From the definition~(\ref{def_Rn}) and from~(\ref{defn}) it follows that one always has
\begin{eqnarray}
{\cal R}_{kn}({\cal A};\hat{H}) \geq {\cal R}_{n}({\cal A};\hat{H}) \quad \forall k,n \in \mathbb{N}\;. 
\end{eqnarray}
Combining the above inequality with~(\ref{orbd}) we have 
\begin{eqnarray}
{\cal R}_n({\cal A};\hat{H})\geq \left\lfloor \frac{n}{K_{\mathbb{S}}+1} \right\rfloor\frac{E M_\mathbb{S} }{n} \quad \forall n > K_{\mathbb{S}} \; ;
\end{eqnarray}
which also implies that, for every $\frac{1}{2} \leq c < 1$, we can write
\begin{eqnarray}
{\cal R}_n({\cal A};\hat{H})\geq c \frac{E M_\mathbb{S} }{K_{\mathbb{S}} + 1}  \quad \forall n > \tfrac{c}{1-c} K_{\mathbb{S}} \; .
\end{eqnarray}
As in the case of~(\ref{harmonic_mean}), for non degenerate spectra this can also be casted in the weaker (yet simpler) form  
\begin{eqnarray}
{\cal R}_n({\cal A};\hat{H}) \geq c \left( \sum_{i \in \mathbb{S} } \frac{1}{\epsilon_i} \right)^{-1}\;,   \qquad \forall n > \tfrac{c}{1-c} K_{\mathbb{S}} \; .
\label{bound_mi2}
\end{eqnarray}

\section{Generic spectra}
\label{sec:spettri_irrazionali}

\begin{figure}
	\includegraphics[width=\columnwidth]{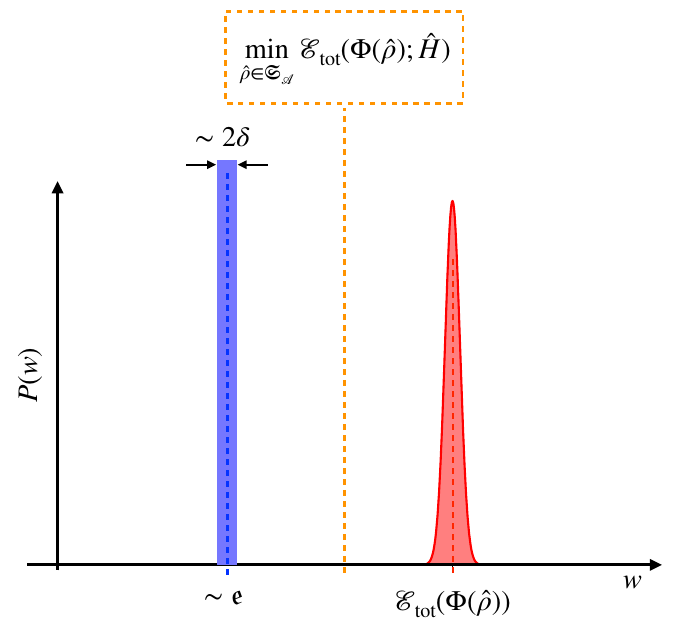}
	\caption{Comparison between the bounded fluctuation protocol presented here and the collective protocol of Ref.~\cite{Perarnau2018}. The typicality-based protocol of~\cite{Perarnau2018} can extract the maximal energy $\mathcal{E}_{\mbox{tot}}(\Phi(\hat\rho))$, with exponentially suppressed fluctuations. In contrast, the protocol discussed in this section can extract an energy $\sim\mathfrak{e} < \mathcal{E}_{\rm tot}(\Phi(\hat\rho))$, but with the guarantee that the work fluctuation never exceeds 4$\delta$. Protocols with bounded fluctuation may exist for higher value of the mean extracted work, up until the upper bound~\ref{newbound11}, which is always smaller or equal than $\mathcal{E}_{\rm tot}(\Phi(\hat\rho))$ (with the equality holding only in the case in which $\hat\rho = \hat\omega_{\mathcal{A}}(\beta^\star)$). }
	\label{w_illustration}
\end{figure}

In this section we are going to show that, by approximating the spectrum of a generic Hamiltonian to rational level, we can construct a work extraction protocol with bounded fluctuations. 
\begin{lemma}
\label{lemma:Varquasi0}
Let $\ham = \sum_{j=0}^{M-1} \epsilon_{j}  \hat{\Pi}_j$ and $\ham^\prime = \sum_{j=0}^{M-1} \epsilon^\prime_{j}  \hat{\Pi}_j$ be two Hamiltonians 
of the form~(\ref{defHAM}) characterized by the same degeneracies values and whose associated  eigenvalues  differ at most by a  constant $\delta \geq 0$, i.e. 
\begin{eqnarray}
\left\lvert \epsilon'_j - \epsilon_j \right\rvert < \delta \;, \qquad \forall j\in \{ 0, \cdots, M-1\}\;. 
\label{energie_vicine}
\end{eqnarray} 
Suppose now that  $\ham^\prime$ admits an unitary evolution $\hat{U}$ that 
permits to extract a deterministic work value $W'\geq 0$ for 
 the 
 input state $\hat{\rho}$, i.e.
 ${P}^{(\hat{H}')}_{\hat{\rho}; \hat{U}} (w) = \delta(w-W')$. 
 Then  
 using  $\hat{U}$ when the system 
  Hamiltonian is $\ham$, yields an average work extraction value 
  \begin{eqnarray} \label{IMPLI11111} 
|\langle W_{\hat{U}}(\hat{\rho}; \hat{H}) \rangle - W'| \leq 2\delta \;, 
\end{eqnarray} 
and a
  probability distribution of the extracted work 
  ${P}^{(\hat{H})}_{\hat{\rho}; \hat{U}} (w)$, which 
   is null whenever the distance of $w$ from $W'$ is larger than $2\delta$, i.e. 
\begin{equation}\label{distrinew} 
{P}^{(\hat{H})}_{\hat{\rho}; \hat{U}} \left( \left\lvert w - W'\right\rvert > 2\delta \right ) :=
1- \int^{W'+2\delta}_{W'-2\delta} d w {P}^{(\hat{H})}_{\hat{\rho}; \hat{U}}(w)= 
 0\;.
\end{equation}
\end{lemma}
\begin{proof}
Invoking the condition~(\ref{def20}) we know that $\hat{U}$ applied to $\hat{\rho}$, given  
$i\in  {\mathbb{S}}[\hat{\rho}]$ 
the transition probabilities $P_{\hat{U}}(j| \hat{\rho}_i)= \mbox{Tr} [\hat{\Pi}_j\hat{U}  \hat{\rho}_i\hat{U}^\dag]$ is equal to one for $j$ such that  $\epsilon'_i-\epsilon'_j=W'$, and
zero otherwise. Observe hence  that according to~(\ref{energie_vicine}), when working with  the Hamiltonian $\hat{H}$ the same mapping will assign
probability  equal to 1  to  energy jumps $\epsilon_i\mapsto \epsilon_j$ which  fulfils the inequality
$|\epsilon_i-\epsilon_j - W'| \leq 2 \delta$, and zero otherwise
\begin{eqnarray} 
P_{\hat{U}}(j| \hat{\rho}_i)= \left\{ \begin{array}{ll} 
1 & \mbox{for} \; |\epsilon_j-\epsilon_i -W'| \leq  2\delta\;,  \\
0 & \mbox{otherwise} \;.
\end{array} \right. 
\end{eqnarray} 
The thesis finally follows by replacing the above identity in Eqs.~(\ref{ddfaver}) and (\ref{distribution}). 
\end{proof}
{\bf Remark:} 
Observe that the above result can also be generalized to situations in which $\hat{H}$ and $\hat{H}'$ commute by have different degeneracies. In particular 
consider the case where $\hat{H}'$ is non-degenerate with eigenvalues $\epsilon'_{0}, \cdots , \epsilon'_{d-1}$, while $\hat{H}$ has only $M < d$  distinct eigenvalues 
$\epsilon_{0}, \cdots , \epsilon_{M-1}$. Assume now that 
we can organize the element of the spectrum of $\hat{H}'$ into $M$ groups 
$\mathbb{G}'_0, \mathbb{G}'_1, \cdots, \mathbb{G}'_{M-1}$ 
such that  
\begin{eqnarray}
\left\lvert \epsilon' - \epsilon_j \right\rvert < \delta \;, \qquad \forall j\in \{ 0, \cdots, M-1\}\;, \forall \epsilon'\in \mathbb{G}'_j\;.
\label{energie_vicine19}
\end{eqnarray} 
Then following the same derivation given in the Lemma can be used to show that 
if  ${P}^{(\hat{H}')}_{\hat{\rho}; \hat{U}} (w) = \delta(w-W')$ for some $\hat{U}$ and $W'$, then Eqs.~(\ref{IMPLI11111}) and~(\ref{distrinew}) still applies.

\begin{lemma}
Given a generic Hamiltonian $\ham = \sum_{i=0}^{M-1} \epsilon_{i}  \hat{\Pi}_i$ and 
${\cal A}:= \bigoplus_{i=0}^{M-1}{\cal A}_{i}$
a direct sum of subsets of its energy eigenspaces, let
\begin{eqnarray}
\mathfrak{e} := \left( \sum_{i \in \mathbb{S}} \frac{\overline{d}_i}{\epsilon_i} \right)^{-1}\;,
\end{eqnarray}
with $\mathbb{S}$ the non-empty elements set~(\ref{defroho}) of ${\cal A}$ and 
with $\overline{d}_i=\mbox{\rm dim}[ {\cal A}_i]$ the dimension of its $i$-th energy block.
Then for each  $\hat{\rho}\in {\mathfrak{S}}_{\cal{A}}$ and 
 $c \in [0,1[$, we can identify a positive constant $A$ with the property that, for each $\delta > 0$ sufficiently small, given  $n > A\delta^{-|\mathbb{S}| + 1}$ copies of $\hat{\rho}$, we can find a TPM protocol acting on $\hat{\rho}^{\otimes n}$ such that
\begin{eqnarray}
&&W:={\langle W_{\hat{U}}(\hat{\rho}^{\otimes n}; \hat{H}^{(n)}/n)\rangle} \geq c\; \mathfrak{e} -2 \delta  \;, \label{condizione_Media} \\ 
&&{P}^{(\hat{H}^{(n)}/n)}_{\hat{\rho}^{\otimes n}; \hat{U}}  \left( \left\lvert w - W \right\rvert > 4\delta \right)= 0 \;.  \label{condizione_distr}
\end{eqnarray}
\end{lemma}
\begin{proof}
Define 
$\{ |\epsilon_{i,j}\rangle : j=0, \cdots, d_i-1\}$ the orthonormal basis of the energy subspace ${\cal H}_i$ of $\hat{H}$ 
constructed by tacking as first $\overline{d}_i$ elements those which define the orthonormal basis of ${\cal A}_i$ introduced in 
Eq.~(\ref{imponew}).   For $\delta >0$ and $i\in\{0,\cdots, M-1\}$ 
 we now introduce
the integer constants
\begin{eqnarray}
m_{i,j} := \left\lfloor \frac{\epsilon_i}{\delta / d^\star} \right\rfloor +
j+1\;, \quad s\in \{ 0, \cdots, d_i-1\}\;, 
\end{eqnarray} 
where $\epsilon_i$ are the eigenvalues of $\hat{H}$ and $d^\star = \max_j d_j$ is its maximum degeneracy.
We hence define 
the Hamiltonian 
\begin{eqnarray} \label{defHHH} 
\ham^\prime :=  \sum_{i=0}^{M-1} \sum_{j=0}^{d_i-1}\epsilon^\prime_{i,j} |\epsilon_{i,j} \rangle \langle \epsilon_{i,j}|\;, \end{eqnarray}
with eigenvalues
\begin{eqnarray} 
\epsilon^\prime_{i,j} := m_{i,j} \frac{\delta}{d^\star}\;.
\end{eqnarray}  which 
by construction  commute with $\hat{H}$ and ${\cal A}$. Notice also that we have 
\begin{eqnarray}\label{newnal} 
\epsilon_i +\delta \geq  \epsilon^\prime_{i,j}  \geq  \epsilon_i\;, \end{eqnarray} 
for all   $i\in \{ 0, \cdots, M-2\}$ and for all $j\in \{0, \cdots, d_i-1\}$. Furthermore if we take
\begin{eqnarray} 
\delta < \min_{i\in \{ 0,\cdots, M-2\}} \left(\epsilon_{i+1}- \epsilon_{i}\right)\;, 
\end{eqnarray} 
 from~(\ref{newnal}) ensures that $\epsilon_{i+1} > \epsilon_{i} +\delta > \epsilon^\prime_{i,j}$ implying that 
 the spectrum of $\ham^\prime$ is non-degenerate.
 Notice also that the subset ${\cal A}$ can be expressed as a direct sum of energy subspaces of $\hat{H}'$ with 
 a non-empty index subset $\mathbb{S}'$ identified by the couples $\{ (i,j) : i\in {\mathbb{S}}, j\in \{ 0, \cdots, \overline{d}_i\}\}$.
Notice also that   the Hamiltonian $\ham^\prime$ falls therefore under the hypotheses of Sec.~\ref{sec:rational_spectra}, and we can invoke~(\ref{bound_mi2}) to deduce that, for each $n > \tfrac{c}{1-c} K_{\mathbb{S}'}$ one has
\begin{eqnarray}
{\cal R}_n({\cal A};\hat{H}^\prime) &\geq& c \left( \sum_{(i, 
j)\in\mathbb{S}' } \frac{1}{\epsilon'_{i,j}} \right)^{-1}  
\nonumber \\ 
&\geq& c  \left( \sum_{(i, 
j)\in\mathbb{S}' } \frac{1}{\epsilon_i} \right)^{-1}
 = 
c\; \mathfrak{e}  \;,
\end{eqnarray}
where in the second inequality we use the left-most part of~(\ref{newnal}). 
This means that  there exists an unitary $\hat{U}$ such that 
it allows us to  deterministically extract a work value lager than or equal to $c \; \mathfrak{e}$ from 
$n > \tfrac{c}{1-c} K_{\mathbb{S}'}$ copies of a generic density matrix 
$\hat{\rho}\in{\mathfrak{S}}_{\cal{A}}$, i.e. 
 \begin{eqnarray}   \left\{ \begin{array}{l} W':=\langle {W_{\hat{U}}(\hat{\rho}^{\otimes n}; \hat{H'}^{(n)}/n)}\rangle  \geq c \; \mathfrak{e} 
 \;,  \\ \langle \Delta^2 W_ {\hat{U}}(\hat{\rho}^{\otimes n}; \hat{H'}^{(n)}/n) \rangle =0\;. 
 \end{array} \right. \label{dfsfsf1001} 
  \end{eqnarray} 
Observe next that the $n$-copy Hamiltonians ${\ham^{(n)}}/n$ and ${\ham^{\prime (n)}}/n$ have eigenvalues 
\begin{eqnarray} \nonumber 
{\epsilon_{\vec{i}}}/{n}&:=& \sum_{s=1}^n 
\frac{\epsilon_{{i}_s}}{n} \;, \quad \vec{i} \in  \{ 0, \cdots, M-1\}^n\;,  \\ 
 {\epsilon^\prime_{\vec{i},\vec{j}}}/{n} &:=& \sum_{s=1}^n 
\frac{\epsilon'_{{i}_s,{j}_s}}{n} \;, \quad \left\{ \begin{array}{ll}\vec{i} \in  \{ 0, \cdots, M-1\}^n\;, \\ 
\vec{j} \in \mathbb{D}_{\vec{i}} \;, \end{array} \right. \nonumber 
\end{eqnarray} 
 with $\mathbb{D}_{\vec{i}}$ defined implicitly by (\ref{defHHH}), which  
  satisfy the condition~(\ref{energie_vicine19}). Indeed identifying 
 $\mathbb{G}'_{\vec{i}}$ as the set formed by the  elements ${\epsilon^\prime_{\vec{i},\vec{j}}}/{n}$ with $\vec{j}\in \mathbb{D}_{\vec{i}}$, 
 from Eq.~(\ref{newnal}) 
 we get: 
\begin{equation}
\left\lvert \frac{\epsilon^\prime_{\vec{i},\vec{j}}}{n} - \frac{\epsilon_{\vec{i}}}{n} \right\rvert \leq \frac{1}{n} \sum_{s=1}^n \left\lvert \epsilon^\prime_{i_s,j_s} - \epsilon_{i_s} \right\rvert \leq \frac{1}{n} \sum_{s=1}^n \delta = \delta \; .
\end{equation}
Applying  the identity~(\ref{IMPLI11111}) of Lemma~\ref{lemma:Varquasi0} to the pair of Hamiltonians $\ham^{(n)}$  and $\ham^{\prime (n)}$ ensures therefore that the same unitary  $\hat{U}$ that realizes~(\ref{dfsfsf1001}), also 
fulfils
\begin{eqnarray} 
|W-W'| \leq 2\delta \;, \label{ddsdsds} 
\end{eqnarray} 
that implies 
~(\ref{condizione_Media}). 
Similarly from~(\ref{distrinew}) of the Lemma we get 
\begin{eqnarray}
0&=&{P}^{(\hat{H}^{(n)}/n)}_{\hat{\rho}^{\otimes n}; \hat{U}}  \left( \left\lvert w - W' \right\rvert > 2\delta \right) \nonumber \\
&=& {P}^{(\hat{H}^{(n)}/n)}_{\hat{\rho}^{\otimes n}; \hat{U}}  \left( \left\lvert (w - W)+( W' -W)\right\rvert > 2\delta \right) \;, \nonumber 
\end{eqnarray} 
which together with (\ref{ddsdsds}) leads to 
~(\ref{condizione_distr}).
To complete the proof, we observe that the constant $K_{\mathbb{S}'}$ can be bounded with
\begin{eqnarray}
K_{\mathbb{S}'} &<& \sum_{i ,j} \frac{M_{\mathbb{S}'}}{m_{i,j}} < \prod_{i,j} m_{i,j} \sum_{i,j} \frac{1}{m_{i,j}} \\ \nonumber 
&<&|\mathbb{S}'| \left( \max_{i,j} m_{i,j} \right)^{|\mathbb{S}'| - 1} \leq |\mathbb{S}'| \left( \frac{d^\star \epsilon_{d-1}}{\delta} +d^\star \right)^{|\mathbb{S}'| - 1} \nonumber \\
&\leq &d^{\star}  |\mathbb{S}  | \left( \frac{d^\star \epsilon_{d-1}}{\delta} +d^\star \right)^{d^{\star} |\mathbb{S}| - 1} \nonumber 
\\
&= &d^{\star}  |\mathbb{S}  | \left({d^\star \epsilon_{d-1}}+d^\star \delta  \right)^{d^{\star} |\mathbb{S}| - 1} \; \delta^{-d^{\star} |\mathbb{S}| + 1}\;,
\end{eqnarray}
where in the third line we invoke the monotonicity under $|\mathbb{S}'|$ and the inequality 
$|\mathbb{S}'|\leq d^{\star}|\mathbb{S}|$. To identify  the constant $A$ finally observe that for sufficiently small $\delta$ we can also write 
\begin{eqnarray}
K_{\mathbb{S}'} < d^{\star}  |\mathbb{S}  | \left({d^\star \epsilon_{d-1}}+1  \right)^{d^{\star} |\mathbb{S}| - 1} \; \delta^{-d^{\star} |\mathbb{S}| + 1}\;,
\end{eqnarray}
which gives the thesis by taking  
\begin{eqnarray} A= \tfrac{c}{1-c} d^{\star}  |\mathbb{S}  | \left({d^\star \epsilon_{d-1}}+1  \right)^{d^{\star} |\mathbb{S}| - 1} \;.\end{eqnarray} 
\end{proof}

\section{A semi-heuristic estimation based on the Central Limit Theorem}
\label{sec:CLT_estimation}

\begin{figure}
	\includegraphics[width=\columnwidth]{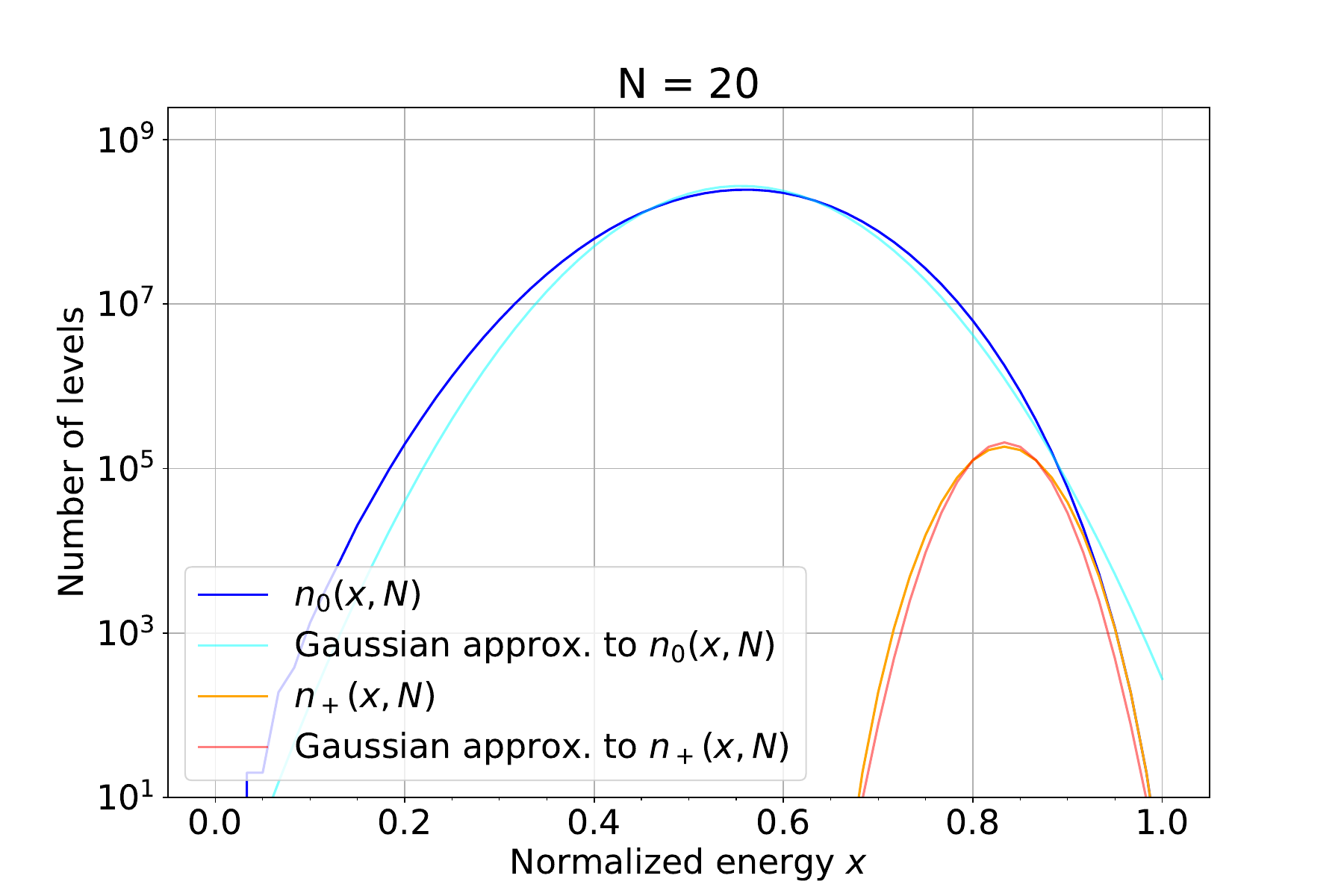}
	\caption{The distribution of energy levels~(\ref{n_zero}) and~(\ref{n_plus}),  compared with their Gaussian estimations~(\ref{n_zero_gaussian}) and~(\ref{n_plus_gaussian}), for $n=20$ copies of the 3-level Hamiltonian $\ham = \frac{2}{3}\ket{1}\bra{1} + \ket{2}\bra{2}$, when $\mathcal{A} = \mbox{Span}\{ \ket{1}, \ket{2} \}$. }
	\label{fig:gaussiane}
\end{figure}

Let $\ham$ be an Hamiltonian satisfying the hypotheses of Sec.~\ref{sec:rational_spectra}. We define the integer quantities $\mathfrak{n}_0(x, n)$ and $\mathfrak{n}_+(x, n)$ as the number of energy levels  with energy equal to $nx$ in, respectively, $\hat{H}^{(n)}$ and $\mathcal{A}^{\otimes n}$. Explicitly,
\begin{eqnarray} 
\label{n_zero}
\!\!\mathfrak{n}_0(x, n) &:=& \#\left\{ \vec{i} \in [0,d-1]^{n} \colon \sum_{s=1}^n \epsilon_{i_s} = nx \right\}, \\
\label{n_plus}
\!\!\mathfrak{n}_+(x, n)&:=&\#\left\{ \vec{i} \in \mathbb{S}^{n} \colon \sum_{s=1}^n \epsilon_{i_s} = nx \right\},
\end{eqnarray}
where we are using the symbol $\#$ to denote the cardinality of a set. Then 
the evaluation of~(\ref{impo11}) can be reformulate as
\begin{equation}
\mathcal{R}_n(\mathcal{A}, \hat{H}) =  \max\left\{ \delta : \forall x \;  \mathfrak{n}_+(x, n) \leq \mathfrak{n}_0(x-\delta, n)   \right\} \; .
\label{criterio_Wdet}
\end{equation}
We observe incidentally that $\mathfrak{n}_0(x, n)$ and $\mathfrak{n}_+(x, n)$ can be expressed as polynomial coefficients in the expansions
\begin{eqnarray}\nonumber 
\left( \sum_{j=0}^{d-1} z^{\epsilon_j} \right)^n = \sum_{i=0}^{nd - 1} \mathfrak{n}_0 \left( \frac{i}{n} , n  \right) z^i \;, \\
\left( \sum_{j \in \mathbb{S}} z^{\epsilon_j} \right)^n = \sum_{i=0}^{nd - 1} \mathfrak{n}_+ \left( \frac{i}{n} , n  \right) z^i \; ,
\end{eqnarray}
which allow for their efficient numerical computation.
Let hence define 
\begin{eqnarray}
\mu_0 &:=& \frac{1}{d} \sum_{i=0}^{d-1} \epsilon_d\;,  \qquad  \sigma^2_0 := \frac{1}{d} \sum_{i=0}^{d-1} \epsilon^2_i - \mu^2_0 \;,\nonumber \\
\mu_+ &:=& \frac{1}{|\mathbb{S}|} \sum_{i \in \mathbb{S} } \epsilon_d \;, \qquad \!\!\sigma^2_+ := \frac{1}{|\mathbb{S}|} \sum_{i \in \mathbb{S}} \epsilon^2_i - \mu^2_+  \;. 
\end{eqnarray}
For large enough $n$, the Central Limit Theorem allows to approximate the energy level densities as
\begin{eqnarray}
\label{n_zero_gaussian}
\mathfrak{n}_0(x, n) &\simeq& \frac{d^n}{n}\frac{\sqrt{n}}{\sqrt{2\pi \sigma^2_0}} \exp\left[ - \frac{n(x-\mu_0)^2}{2\sigma^2_0} \right],  \\
\label{n_plus_gaussian}
\!\!\mathfrak{n}_+(x, n) &\simeq& \frac{|\mathbb{S}|^n}{n}\frac{\sqrt{n}}{\sqrt{2\pi \sigma^2_+}} \exp\left[ - \frac{n(x-\mu_+)^2}{2\sigma^2_+} \right], 
\end{eqnarray}
see  Fig.~\ref{fig:gaussiane} for an illustrative example. 
Exploiting Eqs.~(\ref{n_zero_gaussian}) and (\ref{n_plus_gaussian}), 
 the condition $\forall x \; \mathfrak{n}_+(x, n) \leq \mathfrak{n}_0(x-\delta, n) $ from~(\ref{criterio_Wdet}) becomes
\begin{eqnarray}
 &&\tfrac{|\mathbb{S}|^n}{\sqrt{2\pi \sigma^2_+}} \exp\left[ - \tfrac{n(x-\mu_+)^2}{2\sigma^2_+} \right] 
 \\ && \qquad\qquad \leq\nonumber
\tfrac{d}{\sqrt{2\pi \sigma^2_0}} \exp\left[ - \tfrac{n(x-\delta-\mu_0)^2}{2\sigma^2_0} \right] \; ,
\quad \forall x 
\end{eqnarray}
which with some simple algebraic manipulation can be cast in the form
\begin{eqnarray}\label{quadratic_form}
&&x^2\left( \tfrac{1}{2 \sigma^2_+} - \tfrac{1}{2 \sigma^2_0} \right) - 2x\left( \tfrac{\mu_+}{2\sigma^2_+} - \tfrac{\mu_0 + \delta}{2\sigma^2_0} \right)
 \\
&& \qquad + \left(  \tfrac{\mu^2_+}{2\sigma^2_+} - \tfrac{(\mu_0 + \delta)^2}{2\sigma^2_0} + \ln \tfrac{d}{ | \mathbb{S} |  }  + \tfrac{1}{n}\ln\tfrac{\sigma_+}{\sigma_0} \right) \geq 0 \; . \quad \forall x \nonumber 
\end{eqnarray}
The discriminant  of the above quadratic form is equal to
\begin{equation}
\Delta = \frac{(\mu_0 - \mu_+ + \delta)^2}{\sigma^2_+\sigma^2_0} - \left( \ln \tfrac{d}{|\mathbb{S}|  }  + \tfrac{1}{n}\ln\tfrac{\sigma_+}{\sigma_0} \right) \frac{2\sigma^2_0 - 2\sigma^2_+}{\sigma^2_+\sigma^2_0} \; .
\label{discriminante}
\end{equation}
In order for the condition~(\ref{quadratic_form}) to hold true, we need that $\sigma^2_0 > \sigma^2_+$ and that $\Delta \leq 0$. Solving~(\ref{discriminante}) for $\delta$, we find that the requirement $\Delta \leq 0$ is equivalent to
\begin{equation}
\delta \leq \mu_+ - \mu_0 + \sqrt{2\left( \sigma^2_0 - \sigma^2_+ \right)} \sqrt{\left( \ln \tfrac{d}{|\mathbb{S} |  }  + \tfrac{1}{n}\ln\tfrac{\sigma_+}{\sigma_0}  \right)} \;, 
\end{equation}
which, in the $n \to \infty$ limit, leads to the estimation
\begin{eqnarray}
\label{CLT_estimation}
\mathcal{R}(\mathcal{A}, \hat{H} )  \simeq \mu_+ - \mu_0 + \sqrt{2\left( \sigma^2_0 - \sigma^2_+ \right)} \sqrt{\ln \tfrac{d}{|\mathbb{S}|  } } \;. 
\end{eqnarray}

\section{Conclusions}
\label{sec:conclusions}

We have derived upper and lower bounds on the asymptotic maximal deterministic work extraction (MDEW) rate, which quantifies the maximal work that can be extracted from a quantum system, without fluctuations, in the limit of infinite copies of the system.
We found a lower bound that is strictly greater than zero for any Hamiltonian with rational spectra, meaning that, given enough copies of the system, deterministic work extraction is always possible for such Hamiltonians. Numerical evidence suggests that the actual MDEW rate may coincide with, or be very close to, the upper bound we derived, but we were not able to prove this definitively.

For Hamiltonians with incommensurable energy levels, although strictly deterministic work extraction may not be achievable, we have shown that with enough copies it is possible to bound the fluctuations in the extracted work to an arbitrarily small tolerance. Our protocols for bounded-fluctuation work extraction may find applications in quantum heat engines or batteries where a reliable, stable work output is critical.

More broadly, the scheme that we have introduced for manipulating ensembles of non-interacting copies of a quantum system may have implications for bounding fluctuations of other quantities through global quantum operations on multiple copies. This could aid in the design of stable quantum devices functioning in the finite-copies regime. An open question is whether allowing interactions between copies can enhance deterministic work extraction yields beyond the independent-copies bounds we have derived.

We acknowledge financial support by MUR (Ministero dell’Istruzione, dell’Universit\`a e della Ricerca) through the  PNRR MUR project PE0000023-NQSTI.

\clearpage
\bibliography{Wdet_references}
\appendix

\section{Optimality of  Eq.~(\ref{newbound2})} \label{deri} 
Here we show that at least for those cases where the restriction of $\hat{H}$ over the subspace ${\cal A}$ is not degenerate, then 
\begin{eqnarray} \label{newbound11min} 
\min_{\Phi(\hat{\rho})\in {\mathfrak{S}}_{\cal{A}}}  {\cal E}_{\rm tot} (\Phi(\hat{\rho}); \hat{H})=  \min_{\beta >0}  {\cal E}_{\rm tot} (\hat{\omega}_{\cal A}(\beta); \hat{H})\;. 
\end{eqnarray} 
To see this recall first that the total ergotropy of a generic state $\hat{\rho}$ corresponds to 
\begin{eqnarray}
{\cal E}_{\rm tot} (\hat{\rho}; \hat{H}) = \mbox{Tr}[ \hat{\rho} \hat{H}] - \mbox{Tr}[ \hat{\tau}_{\beta(\hat{\rho})} \hat{H}]\;, 
\end{eqnarray} 
with $\hat{\tau}_{\beta(\hat{\rho})}$ the thermal Gibbs state whose  inverse temperature
$\beta(\hat{\rho})$ 
is fixed in order to ensure that von Neumann entropy of 
 such state equal the one of  $\hat{\rho}$, i.e. 
 \begin{eqnarray}
 S(\hat{\tau}_{\beta(\hat{\rho})})=S(\hat{\rho})\;.
 \end{eqnarray}  
Observe next that if $\hat{H}$  is not degenerate, the entropy of the  Gibbs-like density matrices $\hat{\omega}_{\cal A}(\beta)$ span continuously
from $0$ (for $\beta\rightarrow  \infty$) to $\ln \mbox{Tr}[ \hat{\Pi}_{\cal A}]$  (for $\beta \rightarrow 0$) which is the maximum value allowed for states with support in ${\cal A}$. Given hence 
 $\Phi(\hat{\rho})$ a diagonal ensemble in ${\mathfrak{S}}_{\cal{A}}$,  we can always find  $\beta^\star$ such that the Gibbs-like density matrix $\hat{\omega}_{\cal A}({\beta^\star})$ has entropy equal to then one of $\Phi(\hat{\rho})$.  In such case $\hat{\tau}_{\beta(\Phi(\hat{\rho}))}$ and 
 $\hat{\tau}_{\beta(\hat{\omega}_{\cal A}({\beta^\star}))}$ will match allowing us to write 
 \begin{eqnarray}
{\cal E}_{\rm tot} (\Phi(\hat{\rho}); \hat{H}) &=& \mbox{Tr}[ \Phi(\hat{\rho})\hat{H}] - \mbox{Tr}[ \hat{\tau}_{\beta(\hat{\rho})} \hat{H}]\nonumber \\
&=& \mbox{Tr}[ \Phi(\hat{\rho})\hat{H}] - \mbox{Tr}[ \hat{\tau}_{\beta(\hat{\omega}_{\cal A}({\beta^\star}))}\hat{H}]\nonumber \\
&\geq& \mbox{Tr}[ \hat{\omega}_{\cal A}({\beta^\star})\hat{H}] - \mbox{Tr}[ \hat{\tau}_{\beta(\hat{\omega}_{\cal A}({\beta^\star}))}\hat{H}]\nonumber \\ &=& {\cal E}_{\rm tot} (\hat{\omega}_{\cal A}({\beta^\star}); \hat{H})\;, \label{im} 
 \end{eqnarray} 
 where the last inequality follows from the fact that $\hat{\omega}_{\cal A}({\beta^\star})$ is the state with the minimal energy among those which have the same support and the same entropy, so that 
  \begin{equation} \label{questaqui} 
\mbox{Tr}[ \Phi(\hat{\rho})\hat{H}] \geq \mbox{Tr}[ \hat{\omega}_{\cal A}({\beta^{\star}})\hat{H}] \;.
\end{equation} 
 To see this last fact observer that for $\beta$ arbitrary, invoking the Klein inequality we can write 
 \begin{eqnarray} 
0&\leq&  S(\Phi(\hat{\rho})\| \hat{\omega}_{\cal A}({\beta})) = -S(\Phi(\hat{\rho})) - \mbox{Tr}[ \Phi(\hat{\rho}) \ln \hat{\omega}_{\cal A}({\beta})] \nonumber \\
&=&  -S(\Phi(\hat{\rho}))  + \beta \mbox{Tr}[ \Phi(\hat{\rho}) \hat{H}] + \ln Z_{\cal A}(\beta)\nonumber \\
&=&  -S(\Phi(\hat{\rho}))  + \beta \mbox{Tr}[ \hat{\omega}_{\cal A}({\beta})\hat{H}] + \ln Z_{\cal A}(\beta) \label{impo111}  \\
&&+ \beta(\mbox{Tr}[ \Phi(\hat{\rho})\hat{H}] - \mbox{Tr}[ \hat{\omega}_{\cal A}({\beta})\hat{H}] )  \nonumber \\
&=&  -S(\Phi(\hat{\rho}))  + S(\hat{\omega}_{\cal A}({\beta}))  
+ \beta(\mbox{Tr}[ \Phi(\hat{\rho})\hat{H}] - \mbox{Tr}[ \hat{\omega}_{\cal A}({\beta})\hat{H}] ) \;,  \nonumber 
\end{eqnarray} 
where in the second identity we used the fact that $\Phi(\hat{\rho}) = \Phi(\hat{\rho}) \hat{\Pi}_{\cal A}$ to write 
\begin{eqnarray}
\mbox{Tr}[ \Phi(\hat{\rho})  \ln \hat{\omega}_{\cal A}({\beta})]&=&
 \mbox{Tr}[ \Phi(\hat{\rho}) \hat{\Pi}_{\cal A} \ln \left( \tfrac{ \sum_{i\in\mathbb{S}} \hat{\Pi}_{{\cal A}_i} e^{-\beta \epsilon_i}}{Z_{\cal A}(\beta) }\right)] 
 \nonumber \\
 &=& \beta \mbox{Tr}[ \Phi(\hat{\rho}) \hat{\Pi}_{\cal A}  \sum_{i\in\mathbb{S}} \hat{\Pi}_{{\cal A}_i}  \epsilon_i] + \ln Z_{\cal A}(\beta) \nonumber \\
 &=& \beta \mbox{Tr}[ \Phi(\hat{\rho})\hat{H}]  + \ln Z_{\cal A}(\beta) \;. 
 \end{eqnarray} 
 The inequality~(\ref{questaqui}) finally follows from (\ref{impo111}) by simply  
reorganizing the various  terms and taking $\beta =\beta^\star$.
Since~(\ref{im}) applies to all density matrices $\Phi(\hat{\rho})\in {\mathfrak{S}}_{\cal{A}}$ we conclude that  the minimization~(\ref{newbound11})
can be replaced with (\ref{newbound2}) leading to~(\ref{newbound11min}).

\section{Proof of eq.~(\ref{fixed_point})}
\label{sec:fixed_point}
We start by defining the the real functions
\begin{eqnarray}
\mathfrak{E}_0\left( \beta \right) := \Tr[\hat\tau_{\beta}\ham]\;, && \mathfrak{E}_\mathcal{A}\left( \beta \right) := \Tr[\hat\omega_{\mathcal{A}}(\beta)\ham] \;, \\
\mathcal{S}_0\left( \beta \right) :=S(\hat\tau_{\beta} ) \;, && \mathcal{S}_\mathcal{A}\left( \beta \right) := S(\omega_{\mathcal{A}}(\beta)) \;,\\
Z_0(\beta) := \Tr\left[ e^{\beta\ham} \right] \;, &&  Z_{\mathcal{A}}(\beta) := \Tr\left[ \hat\Pi_{\mathcal{A}} e^{\beta\ham} \right] \;,\\ \nonumber
\end{eqnarray}
with $\tau_{\beta}$ and $\hat\omega_{\mathcal{A}}(\beta)$ as in Eqs.~(\ref{fixed_point}) and (\ref{GIBBSlike}), respectively. The above functions satisfy the relationships
\begin{eqnarray}\label{SvsE}
\mathcal{S}_0 (\beta) &=& \beta \mathfrak{E}_0(\beta) + \ln Z_0(\beta) \;, \\
\mathcal{S}_{\mathcal{A}}
(\beta) &=& \beta \mathfrak{E}_{\mathcal{A}}(\beta) + \ln Z_\mathcal{A}(\beta) \; .
\end{eqnarray}
By deriving~(\ref{SvsE}) with respect to the variable $\beta$ we have the following relations:
\begin{eqnarray}
\label{derivateS}
\frac{\mbox{d} \mathcal{S}_0 }{\mbox{d} \beta} = \beta \frac{\mbox{d} \mathfrak{E}_0 }{\mbox{d} \beta} \;, \qquad
\frac{\mbox{d} \mathcal{S}_{\mathcal{A}} }{\mbox{d} \beta} = \beta \frac{\mbox{d} \mathfrak{E}_{\mathcal{A}} }{\mbox{d} \beta} \; .
\end{eqnarray} 
For fixed $\beta$ define now $\beta^\star$ the inverse temperature such that $S\left( \hat\omega_{\mathcal{A}}(\beta) \right) = S\left( \hat\tau_{\beta^\star} \right)$, i.e. 
\begin{eqnarray}
\beta^\star\left(\beta\right) = \mathcal{S}^{-1}_0 \left( S_{\mathcal{A}} (\beta) \right) \; .
\label{def_betastar}
\end{eqnarray}
Deriving~\ref{def_betastar} and then applying~(\ref{derivateS}) we have that
\begin{eqnarray}
\label{derivata_betastar}
\frac{\mbox{d} \beta^\star }{\mbox{d} \beta} = \frac{\mbox{d} \beta^\star }{\mbox{d} S_\mathcal{A}} \frac{\mbox{d} S_{\mathcal{A}} }{\mbox{d} \beta} = \frac{\mbox{d} \beta^\star }{\mbox{d} S_0} \frac{\mbox{d} S_{\mathcal{A}} }{\mbox{d} \beta} = \frac{\beta}{\beta^\star} \frac{\mbox{d} \beta^\star }{\mbox{d} \mathfrak{E}_0 } \frac{\mbox{d} \mathfrak{E}_{\mathcal{A}} }{\mbox{d} \beta} \; .
\end{eqnarray}
Notice next that the total ergotropy in  the right-hand-side of the upper bound~(\ref{newbound2}) can be expressed as
\begin{eqnarray}
\label{ergotot_fbeta}
{\cal E}_{\rm tot} (\hat{\omega}_{\cal A}(\beta); \hat{H}) = \mathfrak{E}_{\mathcal{A}}(\beta) - \mathfrak{E}_{0}(\beta^\star(\beta)) \; .
\end{eqnarray}

Deriving~(\ref{ergotot_fbeta}) and then using the chain rule and~(\ref{derivata_betastar}) we obtain
\begin{eqnarray}&&
\frac{\mbox{d}}{\mbox{d}\beta} {\cal E}_{\rm tot} (\hat{\omega}_{\cal A}(\beta); \hat{H})  \nonumber \\ 
&&\quad =\frac{\mbox{d}\mathfrak{E}_{\mathcal{A}}}{\mbox{d}\beta} - \frac{\mbox{d}}{\mbox{d}\beta} \mathfrak{E}_{0}(\beta^\star(\beta))= 
\frac{\mbox{d}\mathfrak{E}_{\mathcal{A}}}{\mbox{d}\beta} - \frac{\mbox{d}\mathfrak{E}_{0}}{\mbox{d}\beta^\star} \frac{\mbox{d} \beta^\star }{\mbox{d} \beta}   \nonumber \\ 
&&\quad =\frac{\mbox{d}\mathfrak{E}_{\mathcal{A}}}{\mbox{d}\beta} -
\frac{\beta}{\beta^\star}  \frac{\mbox{d}\mathfrak{E}_{0}}{\mbox{d}\beta^\star} \frac{\mbox{d} \beta^\star }{\mbox{d} \mathfrak{E}_0 } \frac{\mbox{d} \mathfrak{E}_{\mathcal{A}} }{\mbox{d} \beta} = \left( 1 - \frac{\beta}{\beta^\star}  \right) \frac{\mbox{d}\mathfrak{E}_{\mathcal{A}}}{\mbox{d}\beta} \; . \nonumber 
\end{eqnarray}
Every stationary point of ${\cal E}_{\rm tot} (\hat{\omega}_{\cal A}(\beta); \ham )$ must satisfy $\frac{\mbox{d}}{\mbox{d}\beta} {\cal E}_{\rm tot} (\hat{\omega}_{\cal A}(\beta); \hat{H}) = 0$, i.e.
\begin{eqnarray}
\left( 1 - \frac{\beta}{\beta^\star}  \right) \frac{\mbox{d}\mathfrak{E}_{\mathcal{A}}}{\mbox{d}\beta} = 0 \; .
\end{eqnarray}
Since $\frac{\mbox{d}\mathfrak{E}_{\mathcal{A}}}{\mbox{d}\beta} < 0$, we arrive at the conclusion that the bound~(\ref{newbound2}) is attained at a value of $\beta$ such that~
\begin{eqnarray}
 \beta^{\star} \left( \beta \right) = \beta\;,
 \end{eqnarray} 
which proves the thesis.

\end{document}